\documentclass[a4paper,UKenglish,cleveref, autoref, thm-restate]{lipics-v2021}

\usepackage{mathdots}
\usepackage{scalerel}
\usepackage[linesnumbered,ruled,vlined]{algorithm2e} 

\newtheorem{assumption}[theorem]{Assumption}

\hideLIPIcs  

\title{The Guarded Fragment with Nested Equivalences} 

\author{Oskar Fiuk}{Institute of Computer Science, University of Wroc\l{}aw, Poland}{contact.oskarfiuk@gmail.com}{https://orcid.org/0009-0006-1312-4899}{}

\authorrunning{O. Fiuk} 

\Copyright{Oskar Fiuk} 



\category{} 

\relatedversion{} 

\funding{Polish National Science Center, grant No.~2021/41/B/ST6/00996}

\acknowledgements{The author thanks Emanuel Kiero\'nski for helpful discussions, and the reviewers of LICS~2026 for their valuable suggestions.}

\nolinenumbers 

\EventEditors{Claudia Faggian and Joost-Pieter Katoen}
\EventNoEds{2}
\EventLongTitle{41st Annual Symposium on Logic in Computer Science (LICS 2026)}
\EventShortTitle{LICS 2026}
\EventAcronym{LICS}
\EventYear{2026}
\EventDate{July 20--23, 2026}
\EventLocation{Lisbon, Portugal}
\EventLogo{}
\SeriesVolume{380}
\ArticleNo{40}

\begin{document}

\newcommand{\cD}{\mathcal{D}}
\newcommand{\cE}{\mathcal{E}}
\newcommand{\cP}{\mathcal{P}}
\newcommand{\cF}{\mathcal{F}}
\newcommand{\cQ}{\mathcal{Q}}
\newcommand{\cO}{\mathcal{O}}
\newcommand{\cI}{\mathcal{I}}
\newcommand{\cC}{\mathcal{C}}
\newcommand{\cR}{\mathcal{R}}
\newcommand{\cU}{\mathcal{U}}
\newcommand{\cS}{\mathcal{S}}
\newcommand{\cH}{\mathcal{H}}
\newcommand{\cN}{\mathcal{N}}
\newcommand{\cV}{\mathcal{V}}
\newcommand{\cJ}{\mathcal{J}}

\newcommand{\cB}{\mathcal{B}}
\newcommand{\cT}{\mathcal{T}}
\newcommand{\bC}{\mathbf{C}}
\newcommand{\cK}{\mathcal{K}}
\newcommand{\cG}{\mathcal{G}}
\newcommand{\cL}{\mathcal{L}}
\newcommand{\bbP}{\mathbb{P}}
\newcommand{\fA}{\mathfrak{A}}
\newcommand{\fB}{\mathfrak{B}}
\newcommand{\fC}{\mathfrak{C}}
\newcommand{\fD}{\mathfrak{D}}
\newcommand{\fG}{\mathfrak{G}}
\newcommand{\ff}{\mathfrak{f}}
\newcommand{\fg}{\mathfrak{g}}

\renewcommand{\phi}{\varphi} 
\newcommand{\eps}{\varepsilon} 

\newcommand{\TTT}{\mbox{\large \boldmath $\tau$}}
\newcommand{\AAA}{\mbox{\large \boldmath $\alpha$}}
\newcommand{\BBB}{\mbox{\large \boldmath $\beta$}}

\newcommand{\Sat}{\ensuremath{\textit{Sat}}}
\newcommand{\FinSat}{\ensuremath{\textit{FinSat}}}

\newcommand{\FO}{\mbox{\rm FO}}
\newcommand{\FOt}{\mbox{$\mbox{\rm FO}^2$}}
\newcommand{\Ct}{\mbox{$\mathcal{C}^2$}}
\newcommand{\UOF}{\mbox{$\mbox{\rm UF}_1$}}
\newcommand{\ODF}{\mbox{$\mbox{\rm UF}_1$}}
\newcommand{\SUOF}{\mbox{$\mbox{\rm SUF}_1$}}
\newcommand{\RUOF}{\mbox{$\mbox{\rm RUF}_1$}}
\newcommand{\GFt}{\mbox{$\mbox{\rm GF}^2$}}
\newcommand{\GF}{\mbox{$\mbox{\rm GF}$}}
\newcommand{\GFTG}{\mbox{$\mbox{\rm GF+TG}$}}
\newcommand{\GFU}{\mbox{$\mbox{\rm GFU}$}}
\newcommand{\FF}{\mbox{$\mbox{\rm FF}$}}
\newcommand{\ALC}{$\cal ALC$}
\newcommand{\UF}{\mbox{$\mbox{\rm UF}_1$}}
\newcommand{\FAUF}{\mbox{$\forall\mbox{\rm -UF}$}}
\newcommand{\MK}{\mbox{$\mbox{\rm K}$}} 
\newcommand{\MDK}{\mbox{$\mbox{\rm DK}$}} 
\newcommand{\DMDK}{\mbox{$\overline{\mbox{\rm DK}}$}} 
\newcommand{\UNFO}{\mbox{$\mbox{\rm UNFO}$}}
\newcommand{\GNFO}{\mbox{$\mbox{\rm GNFO}$}}
\newcommand{\TGF}{\mbox{$\mbox{\rm TGF}$}}
\newcommand{\LGF}{\mbox{$\mbox{\rm LGF}$}}

\newcommand{\GC}{\mbox{$\mbox{\rm GC}$}}
\newcommand{\MGC}{\mbox{$\mbox{\rm mGC}$}}
\newcommand{\SC}{\mbox{$\mbox{\rm SC}$}}
\newcommand{\USC}{\mbox{$\mbox{\rm USC}$}}
\newcommand{\GSC}{\mbox{$\mbox{\rm GSC}$}}
\newcommand{\DMK}{\mbox{$\overline{\mbox{\rm K}}$}} 
\newcommand{\GFcp}{\mbox{$\mbox{\rm GF}^{\times_2}$}}
\newcommand{\MGF}{\mbox{$\mbox{\rm mGF}$}}

\newcommand{\LogSpace}{\textsc{LogSpace}}
\newcommand{\NLogSpace}{\textsc{NLogSpace}}
\newcommand{\NP}{\textsc{NPTime}}
\newcommand{\PTime}{\textsc{PTime}}
\newcommand{\PSpace}{\textsc{PSpace}}
\newcommand{\ExpTime}{\textsc{ExpTime}}
\newcommand{\ExpSpace}{\textsc{ExpSpace}}
\newcommand{\NExpTime}{\textsc{NExpTime}}
\newcommand{\TwoExpTime}{2\textsc{-ExpTime}}
\newcommand{\TwoNExpTime}{2\textsc{-NExpTime}}
\newcommand{\ThreeNExpTime}{3\textsc{-NExpTime}}
\newcommand{\APSpace}{\textsc{APSpace}}
\newcommand{\AExpSpace}{\textsc{AExpSpace}}
\newcommand{\TOWER}{\textsc{Tower}}
\newcommand{\RE}{\textsc{RE}}

\newcommand{\str}[1]{{\mathfrak{#1}}}
\newcommand{\restr}{\!\!\restriction\!\!}

\newcommand{\N}{{\mathbb N}}   
\newcommand{\Q}{{\mathbb Q}}   
\newcommand{\Z}{{\mathbb Z}}   

\newcommand{\sss}{\scriptscriptstyle}

\newcommand{\tp}{{\rm tp}}
\newcommand{\type}[2]{{\rm tp}^{{#1}}[{#2}]}
\newcommand{\classtype}[2]{{\rm ctp}^{{#1}}[{#2}]}
\newcommand{\classsort}[2]{{\rm sort}^{{#1}}[{#2}]}
\newcommand{\classspectrum}[2]{{\rm spec}^{{#1}}[{#2}]}
\newcommand{\absclass}[3]{E^{{#1}}_{{#2}}[{#3}]}
\newcommand{\tet}[2]{{\mathfrak{t}({#1},{#2})}}

\newcommand{\rs}{\bar{r}}
\newcommand{\vs}{\bar{v}}
\newcommand{\xs}{\bar{x}}
\newcommand{\ys}{\bar{y}}
\newcommand{\zs}{\bar{z}}
\newcommand{\cs}{\bar{c}}

\newcommand{\Cons}{{\rm{Cons}}}
\newcommand{\FreeVars}{{\rm fv}}

\newcommand{\eqdef}{:=}

\let\oldemptyset\emptyset
\let\emptyset\varnothing

\newcommand{\RN}[1]{%
  \textup{\uppercase\expandafter{\romannumeral#1}}%
}

\newcommand{\num}{\mathfrak{n}}
\newcommand{\EQ}{\ensuremath{{\mathcal{EQ}}^{\scaleto{\subseteq}{5pt}}}}
\newcommand{\KEQ}[1]{\ensuremath{{{#1}\text{\rm{-}}\mathcal{EQ}}^{\scaleto{\subseteq}{5pt}}}}

\newcommand{\bigmid}{~\big|~}

\maketitle

\begin{abstract}
  The Guarded Fragment (\GF{}) is a well-established decidable fragment of first-order logic.
  We study an extension of \GF{} with \emph{nested equivalence relations}, namely a family of distinguished binary predicates $E_1, E_2, \dots$ interpreted as equivalence relations such that $E_{k+1}$ is coarser than $E_k$ for every~$k$.
  We show that the equality-free \GF{} with nested equivalence relations enjoys the finite model property and has a decidable satisfiability problem.
  Moreover, we establish tight complexity bounds for satisfiability: \TOWER{}-completeness in general, and $(K{+}2)$-\ExpTime{}-completeness when the number of distinguished predicates is fixed to~$K$.
  Finally, we show that satisfiability becomes undecidable if either the nesting condition is dropped (already with two equivalence relations) or equality is admitted (already with a single equivalence relation).
\end{abstract}

\section{Introduction}\label{sec:intro}

Many reasoning tasks reduce to satisfiability problems over domains whose elements carry data values drawn from potentially infinite or very large ranges.
To model situations in which these values can be compared at different levels of granularity, one may introduce a family of increasingly coarser equivalence relations: two elements are related by the $k$-th relation if their associated values agree at the $k$-th level of precision, where precision decreases monotonically as~$k$ increases.
We refer to such families of equivalence relations as \emph{nested}.

For example, Geographical Information Systems often organise data using progressively broader categories: 
city, state, region, and country.  
This hierarchy is captured by four equivalence relations $E_1$, $E_2$, $E_3$, and $E_4$, 
where $E_1$ relates locations within the same city, $E_2$ within the same state, 
$E_3$ within the same region, and $E_4$ within the same country.  
These equivalence relations are naturally nested: addresses located in the same city are also in the same state, those in the same state are also in the same region, and so on.

Similar hierarchies appear in numerous other contexts: Data Storage (organised into files and folders), Network Management (organised into networks and subnetworks), Dependency Maintenance (organised into modules and submodules), E-Commerce (organised into products, categories, and subcategories).
Given the naturalness and pervasiveness of hierarchically structured data, it is of broad interest to develop reasoning frameworks capable of modelling such data using nested equivalence relations.

\subparagraph*{Previous work.}
Logics with nested equivalences have so far been studied predominantly within the two-variable regime: the Two-Variable Fragment~(\FOt{}) and related formalisms.

One can reason about nested equivalence relations in expressive description logics;
such as Horrocks and Sattler's~\cite{HS99} logic $\mathcal{SHI}$, that is,
the standard $\mathcal{ALC}$ logic enriched with transitive roles ($\mathcal{S}$), role hierarchies ($\mathcal{H}$), and role inverses ($\mathcal{I}$).
Björklund and Bojańczyk~\cite{BB07} investigated \FOt{} over words with \emph{nested data},
that is, structures interpreting a linear order, nested equivalence relations, and further unary predicates.
Recent work of Fiuk, Kiero\'nski, and Michielini~\cite{FKM25} advanced this line of research by establishing decidability for several variants of \FOt{} with nested equivalence relations; see Section~\ref{sec:related} for a detailed comparison.

Although \FOt{} is a prominent decidable fragment, its expressive power becomes severely limited once one needs to reason about more than two elements simultaneously.
To the best of our knowledge, the only decidable fragment supporting unboundedly many variables together with nested equivalence relations that has been systematically studied is an extension of the Unary Negation Fragment, denoted as $\UNFO+\mathcal{SOH}$~\cite{DK19}.

\subsection{Our Contribution}\label{sec:contribution}

In this work we study the Guarded Fragment (\GF{}) equipped with nested equivalence relations.
Up to our knowledge, even the case of \GF{} with a single equivalence relation had remained open.
Previous work on \GF{} considered equivalence relations only in restricted settings, such as \emph{equivalence guards}; see Section~\ref{sec:related} for discussion.

\subparagraph*{Guarded Fragment.}
Introduced by Andr\'{e}ka, van Benthem, and N\'{e}meti~\cite{ABN98} as a generalisation of modal logic,
the Guarded Fragment has become, alongside \FOt{}, one of the central decidable fragments of first-order logic.
The prominence of \GF{} stems from its robust computational and model-theoretic properties.
In particular, \GF{} provides a logical foundation for basic description logics.
For a comprehensive introduction to the Guarded Fragment, see Chapter 4 in Pratt-Hartmann's monograph~\cite{PH23}.

\begin{definition}\label{def:GFsyntax}
The \emph{Guarded Fragment} (\GF) is the set of function-free first-order formulas generated by the following rules:
\begin{romanenumerate}
  \item\label{GF-item1} Every atomic formula is in \GF{}.
  \item\label{GF-item2} \GF{} is closed under Boolean connectives.
  \item\label{GF-item3} If $\psi(x) \in \GF{}$ has at most one free variable, then both $\exists x\,\psi(x)$ and $\forall x\,\psi(x)$ are in \GF{}.
  \item\label{GF-item4} If $\psi(\bar{x}, \bar{y}) \in \GF{}$ and $\gamma(\bar{x},\bar{y})$ is an atom mentioning every variable from $\bar{x} \cup \bar{y}$,
  then both $\exists\bar{x}\,\big(\gamma(\bar{x},\bar{y}) \wedge \psi(\bar{x},\bar{y})\big)$ and $\forall\bar{x}\,\big(\gamma(\bar{x},\bar{y}) \rightarrow \psi(\bar{x},\bar{y})\big)$ belong to \GF{}.
\end{romanenumerate}
The atom $\gamma(\bar{x},\bar{y})$ in rule~(\ref{GF-item4}) is called a \emph{guard}.
\end{definition}

\subparagraph*{Nested equivalences.}
We assume that logical signatures include a family of \emph{distinguished} binary symbols, denoted as $E_1,E_2,E_3,\dots$.
The interpretation of these distinguished symbols is defined such that for every relevant $k \ge 1$ the following holds:
\begin{itemize}
  \item $E_k$ is an equivalence relation.
  \item $E_{k+1}$ is coarser than $E_k$, that is, $\forall x,y\,\big(x E_k y \rightarrow x E_{k+1} y\big)$ is an axiom.
\end{itemize}

Beyond the distinguished symbols, signatures may also include \emph{non-distinguished} symbols of any arity as well as constant symbols; these symbols are \emph{free}, that is, their interpretation is not constrained.
In this work, function symbols of positive arity are not considered.

For $K \in \N$, we denote by $\GF[\KEQ{K}]$ the logic consisting of all formulas of \GF{} 
whose signatures include the first $K$ distinguished symbols $E_1,\dots,E_K$,
and whose models are restricted to structures satisfying the intended semantics of these predicates.
We write $\GF[\EQ]$ for the union of $\GF[\KEQ{K}]$ over all $K \geq 0$.
More generally, for any fragment $\cF$, we use the notation $\cF[\KEQ{K}]$ and $\cF[\EQ]$ as well.
Finally, for $m \in \N$, we denote by $\cF^m$ the fragment of $\cF$ consisting of formulas
with at most $m$ distinct variables.

\subparagraph*{Satisfiability problem.}
In this work, we study the satisfiability problem for $\GF[\EQ]$ and $\GF[\KEQ{K}]$: given a sentence $\phi$, does there exist a model $\str{A}$ such that $\str{A}\models\phi$?

We also investigate whether these fragments enjoy the \emph{finite model property}: namely, whether every satisfiable sentence of a fragment has a finite model.

It is known that \GF{} (without equivalences) has the finite model property, with a doubly exponential upper bound on the size of minimal models~\cite{BGO14}. The satisfiability problem for \GF{} is $2$-\ExpTime{}-complete;
with predicates of bounded arity or finite number of variables, the complexity becomes \ExpTime{}-complete~\cite{Gra99}.
It is also known that $\FOt[\EQ]$ has the finite model property, with an exponential upper bound on the size of minimal models. The satisfiability problem for $\FOt[\EQ]$ is \NExpTime-complete~\cite{FKM25}.
In what follows, we present the results obtained for \GF{}
with nested equivalence relations and at least $3$ variables.

\subparagraph*{Equality.}
We first observe that $\GF^3{}[\KEQ{1}]$---that is, \GF{} with three variables and a single equivalence relation---is undecidable, provided that equality is available in the syntax.

\begin{proposition}\label{thm:GFequality}
  The satisfiability and finite satisfiability problems are undecidable for the constant-free $\GF^3[\KEQ{1}]$.
\end{proposition}
\begin{proof}
  The satisfiability and finite satisfiability problems are undecidable for the G\"odel class with identity, i.e., sentences in the form $\forall x,y\,\exists z\,\psi(x,y,z)$, where $\psi$ is quantifier-free and without constant symbols, but may use equality (see~\cite{Gol84,BGG}).

  A G\"odel-class sentence $\phi = \forall x,y\,\exists z\,\psi(x,y,z)$ embeds into $\GF^3[\KEQ{1}]$ as follows:
  \[
    \forall x,y\;\big(x E_1 y \rightarrow \exists z\; \big(G(x,y,z) \wedge x E_1 z \wedge \psi(x,y,z)\big)\big),
  \]
  where $E_1$ is the distinguished equivalence symbol and $G$ is a fresh symbol serving as a dummy guard.

  The correctness of this reduction is evident:
  since quantifiers are relativised by $E_1$, the structure on every equivalence class is a model of $\phi$.
  For converse, it suffices to take any model of $\phi$ and interpret $E_1$ and $G$ as full relations.
\end{proof}

In the rest of the paper, we therefore focus on the equality-free fragments of \GF{}$[\EQ]$ and \GF{}$[\KEQ{K}]$.
As we will see, excluding equality is sufficient to restore the finite model property and decidability of satisfiability.

\subparagraph*{Lower bound.}

Our first main contribution is to show that both the minimal model size and the complexity of reasoning grow drastically in $\GF^3[\EQ]$.
More precisely, each additional distinguished equivalence relation may contribute an extra exponential layer to the size of minimal models.
For every $K,n \in \N$, we construct a satisfiable sentence $\phi_{K,n}$ in $\GF^3[\KEQ{K}]$ whose length is polynomial in $K$ and~$n$, yet every model of $\phi_{K,n}$ has domain size at least
\[
  \tet{K}{n} = 2^{\iddots^{2^n}},
\]
where the tower of exponentials contains $K$ occurrences of~$2$.

We then use these formulas to implement succinct ``large counters'' in complexity lower-bound constructions, obtaining $(K{+}1)$-\ExpTime-hardness of satisfiability for $\GF^3[\KEQ{K}]$.
Subsequently, we lift this lower bound to $(K{+}2)$-\ExpTime-hardness for $\GF[\KEQ{K}]$, i.e., for the unrestricted-variable fragment with constants.
When the number of distinguished equivalence symbols is unbounded, i.e., in $\GF[\EQ]$, the satisfiability problem becomes non-elementary; more precisely, it becomes \TOWER{}-hard.

\begin{theorem}\label{thm:GF3lower}
  The satisfiability problem for the constant-free, equality-free $\GF^3[\EQ]$ is \TOWER-hard.
  
  For every $K \ge 1$, the satisfiability problem for the constant-free, equality-free $\GF^3[\KEQ{K}]$ is $(K{+}1)$-\ExpTime-hard.

  For every $K \ge 1$, the satisfiability problem for the equality-free $\GF[\KEQ{K}]$ is $(K{+}2)$-\ExpTime-hard.
\end{theorem}

\subparagraph*{Upper bound.}
To establish decidability of satisfiability, we first show that the finite model property holds.
Moreover, we derive tight bounds on the size of minimal models.

\begin{theorem}\label{thm:GF-fmp}
  The equality-free \GF{}$[\EQ]$ has the finite model property.

  For every $K \ge 1$, the equality-free \GF{}$[\KEQ{K}]$ has the finite model property
  with a $(K{+}2)$-exponential upper bound on the size of minimal models.
  This bound decreases to $(K{+}1)$-exponential whenever either
  \begin{itemize}
    \item constants are disallowed, or
    \item the number of variables $m \ge 3$ is fixed.
  \end{itemize}
\end{theorem}

The finite model property alone does not, however, yield complexity upper bounds that match
the lower bounds of Theorem~\ref{thm:GF3lower}.
Indeed, a naïve ``guess-and-verify'' approach would incur nondeterministic complexity
for the fragments $\GF[\KEQ{K}]$ parametrised by~$K$.
We therefore design a dedicated decision procedure for satisfiability and obtain the following result,
complementing Theorem~\ref{thm:GF3lower} with matching upper bounds.

\begin{theorem}\label{thm:GFmain}
  The satisfiability problem for the equality-free \GF{}$[\EQ]$ is \TOWER-complete.

  For every $K \ge 1$, the satisfiability problem for the equality-free \GF{}$[\KEQ{K}]$
  is $(K{+}2)$-\ExpTime-complete.
  The complexity decreases to $(K{+}1)$-\ExpTime-complete whenever either
  \begin{itemize}
    \item constants are disallowed, or
    \item the number of variables $m \ge 3$ is fixed.
  \end{itemize}
\end{theorem}

\subparagraph*{Equivalences without nesting.}
Kiero\'nski and Otto~\cite{KO12} showed that $\GFt{}$ is undecidable when extended with three \emph{independent}---that is, not necessarily nested---equivalence relations, but remains decidable with at most two of them.
We complement this result by showing that $\GF^3{}$ becomes undecidable when extended with just two independent equivalence relations.

\begin{restatable}{theorem}{undecGF}\label{thm:GFundec}
  The satisfiability and finite satisfiability problems are undecidable for the constant-free, equality-free $\GF^3$
  extended with two distinguished binary symbols $E$ and $F$, each interpreted as an equivalence relation, but with no nesting conditions between them.
\end{restatable}

\subsection{Overview of Technical Sections}

\smallskip
\noindent\emph{Section~\ref{sec:prelim}.}
Introduces basic notions and definitions.

\smallskip
\noindent\emph{Section~\ref{sec:appDL}.}
Discusses application of our results in ontology languages.
We propose a decidable extension of the description logic $\mathcal{ALCHI}$ equipped with nested equivalence roles.

\smallskip
\noindent\emph{Section~\ref{sec:GF3-lower}.}
Establishes non-elementary lower bounds on the satisfiability problem (Theorem~\ref{thm:GF3lower}).

\smallskip
\noindent\emph{Section~\ref{sec:GF-upper}.}
Establishes the finite model property and the decidability of satisfiability.
This section also presents the decision procedure for the fragments $\GF[\KEQ{K}]$ parametrised by $K \ge 1$ (Theorems~\ref{thm:GF-fmp} and~\ref{thm:GFmain}).

\smallskip
\noindent\emph{Section~\ref{sec:related}.}
Discusses related work.

\smallskip
\noindent\emph{Section~\ref{sec:remarks}.}
Concludes the paper with observations concerning slight extensions of our results.

\smallskip
\noindent\emph{Appendix.}
Part of the material is delegated to the appendices.
In particular, the undecidability result for $\GF^3$ with two independent equivalences
(Theorem~\ref{thm:GFundec}) is proved in Appendix~\ref{sec:GF-undec}.

\section{Preliminaries}\label{sec:prelim}

We denote the set of natural numbers with $0$ by $\N$.
For $k \in \N$, the notation $[k]$ stands for the set $\{1, \dots, k\}$, with the convention that $[0] = \emptyset$. More generally, we use interval notation $[a, b] \subseteq \N$ to denote the set $\{a, a+1, \dots, b\}$ whenever $a \leq b$, and the empty set $\emptyset$ whenever $a > b$.
We denote by $2^S$ the powerset of a set $S$.
If $E \subseteq S \times S$ is an equivalence relation,
then $\absclass{}{}{a}$ is the equivalence class of $a \in S$,
and $S'/E$ is the quotient set of a subset $S' \subseteq S$. 

For $k \in \N$ and $n \in \N$,
the notation $\tet{k}{n}$ stands for \emph{tetration},
defined inductively: $\tet{0}{n} \eqdef n$ and $\tet{k+1}{n} \eqdef 2^{\tet{k}{n}}$.
Whenever we say that $y$ (a quantity) is \emph{$k$-exponential} in $x$ (a parameter), then we really mean that $y \le \tet{k}{p(x)}$ for some \emph{fixed} polynomial $p$.

\subparagraph*{First-Order Logic.}
We assume general familiarity with First-Order Logic ($\FO$).  
The logical symbols are $=, \bot, \top, \vee, \wedge, \neg, \rightarrow, \leftrightarrow$, and the quantifiers $\forall, \exists$.  
Formulas may also use non-logical symbols: relation symbols (from a countably infinite set), constant symbols (also from a countably infinite set), and variables (again countably many).  
Every relation symbol comes with associated arity.
We do not allow function symbols of positive arity.

Let $\phi$ be a first-order formula.
The \emph{length} of $\phi$, denoted $|\phi|$, is the total number of symbols it contains, where each occurrence of a symbol---be it a variable, relation symbol, or constant---contributes~$1$.
We use $\FreeVars(\phi)$ to denote the set of \emph{free} variables of $\phi$;
and write $\phi(\xs)$ to denote $\FreeVars(\phi) \subseteq \xs$.
A~\emph{signature} $\sigma$ of $\phi$ is the finite set of relation and constant symbols that appear in $\phi$.
For $K \ge 1$, if $\phi \in \GF[\KEQ{K}]$, then we assume that $\{ E_1,\dots,E_K\} \subseteq \sigma$ but $E_\ell \not\in \sigma$ for all $\ell > K$, where $E_1,E_2,\dots$ is the family of distinguished binary predicates.

We use Fraktur letters such as $\str{A}, \str{B}, \dots$ to denote structures, and the corresponding Roman letters $A, B, \dots$ for their domains.
A \emph{$\sigma$-structure} $\str{A}$ interprets the symbols from $\sigma$: a relation symbol $R$ as a relation $R^{\str{A}}\subseteq A^k$ with $k$ denoting the arity of $R$; and a constant symbol $c$ as an element $c^{\str{A}} \in A$.
If $B \subseteq A$, we write $\str{A} ~\restr~ B$ for the \emph{restriction} of $\str{A}$ to the subdomain $B$;
to remain a $\sigma$-structure, $B$ must include all interpretations of constants.

\section{Application in Ontology Languages}\label{sec:appDL}

To illustrate the relevance of our results, we discuss a case study inspired by real-life access control policies, a topic naturally arising in knowledge representation and reasoning.

\subparagraph*{Example (Rule-Based Access Control).}
Consider an IT system involving administrators, users, and documents. 
Each entity is affiliated with exactly one organisation, and every organisation is partitioned into departments.
Administrators may manage documents within their departments, and users may request to download documents. 
To enforce data confidentiality, the following access policy is adopted:
\emph{A user may download a document only if an administrator who manages the document has granted the user permission, provided that the user and the administrator belong to the same organisation.}

The described scenario can be formalised in $\GF[\KEQ{2}]$ as follows.
Let $E_1$ relate entities belonging to the same department, and let $E_2$ relate entities within the same organisation, where $E_2$ is coarser than $E_1$. 
The policy is then captured by the following axioms:
\begin{align}\label{formula:accesspolicy-first}
  \forall a,u\;&\big({\rm grant}(a,u) \rightarrow \big({\rm Admin}(a) \wedge {\rm User}(u)\big)\big) \\
  \forall a,d\;&\big({\rm manages}(a,d) \rightarrow \big({\rm Admin}(a) \wedge {\rm Doc}(d)\big)\big) \\
  \forall u,d\;&\big({\rm access}(u,d) \rightarrow \big({\rm User}(u) \wedge {\rm Doc}(d)\big)\big) \\
  \forall a,u\;&\big({\rm grant}(a,u) \rightarrow a E_2 u\big) \\
  \forall a,d\;&\big({\rm manages}(a,d) \rightarrow a E_1 d\big) \\
  \forall u,d\;&\big({\rm access}(u,d) \rightarrow \exists a\, \big({\rm grant}(a,u) \wedge {\rm manages}(a,d) \big)\big)\label{formula:accesspolicy-last}
\end{align}

In~\eqref{formula:accesspolicy-last} the quantifier $\exists a$ is left unguarded; however, w.r.t. satisfiability, one can introduce a fresh ternary symbol ${\rm Guard}$ and put there ${\rm Guard}(u,d,a)$ to serve as a dummy guard.

\subparagraph*{Extending Description Logics.}
We adopt a common terminology of \emph{description logics} (DLs); for a gentle introduction to the topic, see, e.g.,~\cite{Baader_Horrocks_Lutz_Sattler_2017}.

The sentences~\eqref{formula:accesspolicy-first}--\eqref{formula:accesspolicy-last} use only unary and binary predicates. 
Nevertheless, sentence~\eqref{formula:accesspolicy-last} is expressible neither in \FOt{} nor in standard description logics from the $\mathcal{ALC}$ family, including expressive formalisms such as $\mathcal{SHI}$ that can define nested equivalence relations.
To integrate such properties into the DL framework,
we propose an \emph{ad-hoc} extension of the logic $\mathcal{ALCHI}$.

The standard constructs of $\mathcal{ALCHI}$ on top of $\mathcal{ALC}$ are \emph{role inclusions}~($\mathcal{H}$) and \emph{role inverses}~($\mathcal{I}$)---both already expressible in the equality-free \GFt{}. 
The unique features added here are \emph{nested equivalence roles} and \emph{pairwise existential dependencies}.

\emph{Nested equivalence roles.}
The set of role names contains a distinguished family $E_1,E_2,\dots$.
In every interpretation $\cI = (\Delta^{\cI},\cdot^{\cI})$, these roles are interpreted so that:
\begin{itemize}
  \item for every $k \geq 1$, $E_k^{\cI}$ is an equivalence relation over $\Delta^{\cI}$;
  \item for every $k \geq 1$, $E_{k+1}^{\cI}$ is coarser than $E_k^{\cI}$, i.e., \(E_k^{\cI} \subseteq E_{k+1}^{\cI}.\)
\end{itemize}

\emph{Pairwise existential dependencies.}
Let $r$, $p$, $q$ be role names or their inverses.
A \emph{pairwise existential dependency} is an axiom of the form
\begin{equation}\label{eqn:pairwise-inclusion-axiom}
  r \sqsubseteq \exists (p,q)
\end{equation}
An interpretation $\cI = (\Delta^{\cI},\cdot^{\cI})$ satisfies~\eqref{eqn:pairwise-inclusion-axiom} whenever, for every pair $(a,b) \in r^{\cI}$, there exists some $c \in \Delta^{\cI}$ such that $(a,c) \in p^{\cI}$ and $(b,c) \in q^{\cI}$. 
Intuitively, such axioms link pairs of $r$-related individuals to a common witness via $p$ and $q$.
This construct goes beyond $\FOt{}$ as well as the core description logics of the $\mathcal{ALC}$ family.
This additional expressivity is necessary to capture many real-life properties such as~\eqref{formula:accesspolicy-last}.

\smallskip
Having the discussed above features, the sentences~\eqref{formula:accesspolicy-first}--\eqref{formula:accesspolicy-last} can be expressed concisely as:
\begin{equation*}
  \begin{array}{rlrlrl}
    {\rm grant}&\sqsubseteq\; {\rm Admin} \times {\rm User} &
    {\rm manages}&\sqsubseteq\; {\rm Admin} \times {\rm Doc} &
    {\rm access}&\sqsubseteq\; {\rm User} \times {\rm Doc} \\
    {\rm grant}&\sqsubseteq\; E_2 &
    {\rm manages}&\sqsubseteq\; E_1 &
    {\rm access}&\sqsubseteq\; \exists ({\rm grant}^{-1},{\rm manages}^{-1}) \\
  \end{array}
\end{equation*}

From our results it follows that the knowledge-base consistency problem
for the outlined extension of $\mathcal{ALCHI}$ is decidable:
given a knowledge base $\mathcal{K}$ (i.e., a finite set of axioms),
the task is to decide whether there exists an interpretation $\mathcal{I}$
such that $\mathcal{I} \models \mathcal{K}$ (i.e., $\mathcal{I}$ satisfies every axiom from $\mathcal{K}$).
In addition to decidability, we obtain that general and finite consistency coincide: the interpretation $\mathcal{I}$ can be chosen to be finite.

\section{Non-Elementary Lower Bound}\label{sec:GF3-lower}

In this section, we first show that the satisfiability problem for the constant-free, equality-free $\GF^3[\KEQ{K}]$ is $(K{+}1)$-\ExpTime-hard (Subsections~\ref{sec:GF3-large-counters} and~\ref{sec:GF3-complexity}).
We then extend the lower-bound construction to obtain $(K{+}2)$-\ExpTime-hardness for equality-free $\GF[\KEQ{K}]$ with constants and unboundedly many variables (\autoref{sec:lifting-lower-bounds}).
This establishes Theorem~\ref{thm:GF3lower}.

\subsection{Implementing Large Counters}\label{sec:GF3-large-counters}

Fix $K, n \in \N$. For each $k \in [K]$, an integer in the range $[0,\tet{k}{n}-1]$ is called a \emph{$k$-numeral}. 

As a first step in the lower-bound construction, we define a sentence $\phi_{K,n}$ in the constant-free, equality-free $\GF^3[\KEQ{K}]$ that describes a counter over $\tet{K}{n}$ elements.
The sentence $\phi_{K,n}$ is given as the conjunction
\(\phi_{K,n}^{(1)} \land \phi_{K,n}^{(2)} \land \cdots \land \phi_{K,n}^{(K)},\)
realising the following idea.
\begin{itemize}
  \item Using $n$ unary predicates, the sentence $\phi_{K,n}^{(1)}$ assigns to each element a $1$-numeral, and enforces that every $E_1$-class contains $2^n$ elements (at least one for each $1$-numeral).
  
  \item Treating domain elements as bits, the sentence $\phi_{K,n}^{(2)}$ assigns to each $E_1$-class a $2$-numeral, and requires every $E_2$-class to contain $2^{2^n}$ distinct $E_1$-classes (one for each $2$-numeral).
  
  \item Treating entire $E_1$-classes as bits, the sentence $\phi_{K,n}^{(3)}$ assigns to each $E_2$-class a $3$-numeral, and requires every $E_3$-class to contain $2^{2^{2^n}}$ distinct $E_2$-classes (one for each $3$-numeral).
  
  \item Proceeding inductively, for each $k \geq 2$, the sentence $\phi_{K,n}^{(k)}$ assigns a $k$-numeral to every $E_{k-1}$-class, and enforces every $E_k$-class to contain $\tet{k}{n}$ distinct $E_{k-1}$-classes.
  
  \item Consequently, every $E_K$-class contains $\tet{K}{n}$ distinct $E_{K-1}$-classes, each representing a different $K$-numeral.
  
  \item Finally, the length of $\phi_{K,n}$ is polynomial in $n$ and $K$; in fact, \( |\phi_{K,n}| = \cO(n + K). \)
\end{itemize}

The described idea of ``nested counters'' is certainly not new, as similar ones have already appeared, e.g., in Stockmeyer's proof concerning $\FO$ on words~\cite{Sto74}, or in the non-elementary lower bound for the Fluted Fragment~\cite{P-HST16}.
The non-trivial part, however, is to express required properties using just the quite limited syntax of $\GF^3$.

\subparagraph*{Signature.}
The signature of $\phi_{K,n}$ consists of the following predicates.
\begin{itemize}
  \item The distinguished nested equivalence symbols $E_1,\dots,E_K$.
  \item An additional symbol $E_0$ to be used as a sentinel.
  A priori it is a free symbol; however, we later axiomatise it to behave like the finest equivalence relation. Thus, we form an extended family of nested equivalence relations $E_0,E_1,\dots,E_{K}$.
  \item Unary symbols $BB_i$ and $BC_i$ for $i \in [0,n - 1]$ to represent bit and carry values of $1$-numerals, respectively.
  \item Unary symbols $B_k$ and $C_k$ for $k \in [2,K]$ to represent bit and carry values of $k$-numerals, respectively.
  \item Binary symbols $S_k$ for $k \in [K]$ to connect (elements of) $E_{k-1}$-classes that represent circularly successive $k$-numerals.
  \item Unary symbols $Z_k$ for $k \in [K]$ to decorate (elements of) $E_{k-1}$-classes that represent the $k$-numeral~$0$.
  \item Binary symbols $Q_k$ for $k \in [K]$ to connect (elements of) $E_{k-1}$-classes that represent equal $k$-numerals.
  \item Additional symbols serving auxiliary roles, e.g., as guards; not listed here.
\end{itemize}

\subparagraph*{Interpretation.}
The sentence $\phi_{K,n}$ will be constructed so that every model $\str{A} \models \phi_{K,n}$ will satisfy the following conditions.

First, the sentinel symbol $E_0$ shall be interpreted as an equivalence relation which is finer than $E^{\str{A}}_1$. In particular, all below quotients via $E_0$ are well-defined.

For every $k \in [K]$ and every $E_k$-class $\cC \in A/E^{\str{A}}_{k}$, we shall have
\[ \big|\cC/E^{\str{A}}_{k-1}\big| = \tet{k}{n}. \]

To formally state the intended interpretation of the predicates $S_k$, $Z_k$, and $Q_k$,
we introduce a family of functions $\num_{k} \colon A \rightarrow [0,\tet{k}{n}{-}1]$ for $k \in [K]$ (i.e., the $k$th function assigns $k$-numerals to domain elements). These functions are defined inductively:
\begin{itemize}
  \item The first function $\num_1$ on element $a \in A$ is given by
  \[ \num_1(a) \eqdef \sum\nolimits_{i=0}^{n-1} b_i \cdot 2^i, \]
  where $b_i = 1$ if $\str{A}\models BB_i(a)$; and otherwise $b_i = 0$.
  \item For $k \in [2,K]$, the function $\num_{k}\colon A \rightarrow [0,\tet{k}{n}-1]$ is given by mapping $a \in A$ to
  \[ \num_{k}(a) \eqdef \sum\nolimits_{i=0}^{\tet{k-1}{n}-1}{b_i \cdot 2^i}, \]
  where $b_i = 1$ if $\str{A}\models B_k(b)$ for some $b \in \absclass{\str{A}}{k-1}{a}$ with $\num_{k-1}(b) = i$; and otherwise $b_i = 0$.
\end{itemize}

Now, for every $k \in [K]$ and every pair $(a,b) \in E^{\str{A}}_K$, the following shall hold:
\begin{itemize}
  \item $\str{A} \models E_{k-1}(a,b)$ \, iff \, $(a,b) \in E^{\str{A}}_{k}$ and $\num_k(a) = \num_k(b)$.
  \item $\str{A} \models S_k(a,b)$ \, iff \, $\num_k(a) + 1 \equiv \num_k(b) \mod \tet{k}{n}$.
  \item $\str{A} \models Z_k(a)$ \, iff \, $\num_k(a) = 0$.
  \item $\str{A} \models Q_k(a,b)$ \, iff \, $\num_k(a) = \num_k(b)$.
\end{itemize}

Note that every $E_K$-class is supposed to describe an independent $K$-exponential counter.
In particular, we do not impose any constraints on pairs $(a,b) \not\in E^{\str{A}}_K$.

\subparagraph*{Convention.}
In the presentation that follows, we occasionally write formulas that do not match Definition~\ref{def:GFsyntax} (the syntax of \GF{}),
but are equivalent to guarded ones modulo simple rewritings.
We also omit guards for existential quantifiers, since one may always introduce fresh symbols to act as dummy guards.

\subparagraph*{Base conjunct.}
We now construct the first sentence~$\phi_{K,n}^{(1)}$.

To work with $1$-numerals, i.e., integers in $[0,2^n{-}1]$,
we introduce auxiliary formulas implementing cyclic successor (i.e., successor modulo $2^n$), zero test, and equality test.

Recall the standard algorithm for incrementing a number $m$ modulo a power of two.
We first locate the least significant zero bit of $m$, say at position $i$,
and then flip every bit $j \le i$.
If all bits of $m$ are already equal to $1$, then every bit is flipped.
Equivalently, the increment operation can be seen as a bitwise XOR with the \emph{carry bits}.
For example, if
\(m = 1010111,\)
then the carry bits are
\(0001111,\)
and the bitwise XOR produces the successor
\(1011000.\)

We implement increments as described above.
Axiomatise first the carry bits:
\begin{flalign}\label{phi0-first}
  \forall x\;\Big(BC_0(x) \; \wedge \; \bigwedge\nolimits_{i=0}^{n-2} \big(BC_{i+1}(x) \leftrightarrow \big(BB_i(x) \wedge BC_i(x)\big)\big)\Big)
\end{flalign}
Define then a formula testing whether two elements represent circularly successive $1$-numerals:
\begin{flalign*}
    \varphi_{S}(x_1,x_2) \; \eqdef \; \bigwedge\nolimits_{i=0}^{n-1} \big(BB_i(x_2) \leftrightarrow \big(BB_i(x_1) \oplus BC_i(x_1)\big)\big)
\end{flalign*}
The formulas for zero and equality tests are straightforward:
\begin{flalign*}
  \varphi_{Z}(x) \; \eqdef \; \bigwedge\nolimits_{i=0}^{n-1} \neg BB_i(x)
  \quad\text{and}\quad
  \varphi_{Q}(x_1,x_2) \; \eqdef \; \bigwedge\nolimits_{i=0}^{n-1}\big(BB_i(x_1) \leftrightarrow BB_i(x_2)\big)
\end{flalign*}

We now axiomatise the symbols $S_1$, $Z_1$, and $Q_1$; and enforce the existence of successors:
\begin{flalign}\label{unary-arithmetic-0}
  \forall x_1,x_2\; &\big(x_1 E_K x_2 \rightarrow \big(S_1(x_1,x_2) \leftrightarrow \varphi_{S}(x_1,x_2)\big)\big) \\
  \forall x\; &\big(Z_1(x) \leftrightarrow \varphi_{Z}(x)\big) \\
  \forall x_1,x_2\; &\big(x_1 E_K x_2 \rightarrow \big(Q_1(x_1,x_2) \leftrightarrow \varphi_{Q}(x_1,x_2)\big)\big)\label{unary-arithmetic-2} \\
  \forall x\;\exists y\; &\big(S_1(x,y) \wedge x E_1 y\big)
\end{flalign}

Since $S_1$ is cyclic,
inside every $E_1$-class all $2^n$ elements will be eventually generated, regardless of the initial value.

We make the sentinel symbol $E_0$ act as the finest equivalence relation:
\begin{flalign}
  \forall x_1,x_2\;& \big(x_1 E_0 x_2 \rightarrow x_1 E_1 x_2\big) \\
  \forall x_1,x_2\;& \big(x_1 E_1 x_2 \rightarrow \big(x_1 E_0 x_2 \leftrightarrow Q_1(x_1,x_2)\big)\big)\label{phi0-last}
\end{flalign}

Define $\phi_{K,n}^{(1)}$ as the conjunction of sentences \eqref{phi0-first}--\eqref{phi0-last}.

\subparagraph*{Inductive conjuncts.}
Let $k \in [2,K$].
Suppose that the sentences
$\phi_{K,n}^{(1)}, \dots, \phi_{K,n}^{(k-1)}$ have already been defined.
We now construct the next sentence $\phi_{K,n}^{(k)}$, which axiomatises the
predicates $S_k$, $Z_k$, and $Q_k$.
As a result, each $E_{k-1}$-class is associated with a $k$-numeral, that is,
an integer in the range $[0, \tet{k}{n} - 1]$.

\smallskip\noindent\emph{Bits and carry bits.}
We treat the entire $E_{k-2}$-classes as individual bits, where the respective bit values are given by the predicate $B_{k}$.
We require that elements of every $E_{k-2}$-class are assigned consistent bit values:
\begin{flalign}\label{phik-first}
  \forall x_1,x_2\;\big(x_1 E_{k-2} x_2 \rightarrow \big(B_k(x_1) \leftrightarrow B_k(x_2)\big)\big)
\end{flalign}

We next axiomatise the carry bits.
First, the carry bit of the least significant position is always on:
\begin{flalign}
  \forall x\;\big(Z_{k-1}(x) \rightarrow C_k(x)\big)
\end{flalign}
We then propagate the carry until we encounter the first bit off or reach the last bit position; since $S_{k-1}$ is cyclic, the latter is detected by $\neg Z_{k-1}(x_2)$:
\begin{flalign}
  \forall x_1,x_2\;\Big(\big(x_1 E_{k-1} x_2{\wedge}S_{k-1}(x_1,x_2){\wedge}\neg Z_{k-1}(x_2)\big) \rightarrow
  \big(C_k(x_2) \leftrightarrow \big(C_k(x_1){\wedge}B_k(x_1)\big)\big)\Big)
\end{flalign}

\smallskip\noindent\emph{Successors.}
We next ensure that whenever $S_k$ connects two elements, then their $E_{k-1}$-classes indeed represent consecutive $k$-numerals.

First we ensure the following property:
whenever $S_k$ connects a pair of elements across two $E_{k-1}$-classes, say $\cC$ and $\cC'$,
then necessarily $S_k$ connects a pair of elements across every $E_{k-2}$-classes $\cD \in \cC/E_{k-2}$ and $\cD' \in \cC'/E_{k-2}$.
This property translates to:
\begin{flalign}
  &\forall x_1,x_2\,\big(S_k(x_1,x_2) \rightarrow
  \exists y\,\big(S_{k-1}(x_1,y) \wedge x_1 E_{k-1} y \wedge S_k(y,x_2)\big)\big) \label{eqn:stinki-axiom-1} \\
  &\forall x_1,x_2\,\big(S_k(x_1,x_2) \rightarrow
  \exists y\,\big(S_{k-1}(x_2,y) \wedge x_2 E_{k-1} y \wedge S_k(x_1,y)\big)\big) \label{eqn:stinki-axiom-2}
\end{flalign}

We next apply the bitwise XOR with the respective carry bits:
\begin{flalign}
  \forall x_1,x_2\; \big(\big(S_k(x_1,x_2) \wedge Q_{k-1}(x_1,x_2)\big) \rightarrow \chi(x_1,x_2)\big)
\end{flalign}
where $\chi(x_1,x_2) \eqdef B_k(x_2) \leftrightarrow \big(B_k(x_1) \oplus C_k(x_1)\big)$.

So far, whenever $S_k$ links any two elements, then their $E_{k-1}$-classes indeed represent successive $k$-numerals.
Though, we did not actually enforced any $S_k$-connections.
To fix this, we require that whenever $S_k$ does not link some two elements, then their $E_{k-1}$-classes do not represent successive $k$-numerals, i.e., there is a position where the respective bit and carry-bit values disagree.
With $4$ variables, this property is expressed by the sentence:
\begin{flalign*}
  \forall x_1,x_2\; \Big(&\big(x_1 E_K x_2 \wedge \neg S_k(x_1,x_2)\big) \rightarrow \nonumber \\
  &\exists y_1,y_2\; \big(x_1 E_{k-1} y_1 \wedge x_2 E_{k-1} y_2 \wedge Q_{k-1}(y_1,y_2) \wedge \neg\chi(y_1,y_2)\big)\Big)
\end{flalign*}
Using a fresh binary symbol $W_k$, we rewrite the above sentence into a pair of three-variable sentences:
\begin{flalign}\label{eqn:three-variable-trik1}
  &\forall x_1,x_2\;\big(\big(x_1 E_K x_2 \wedge \neg S_k(x_1,x_2)\big) \rightarrow
  \exists y_1\;\big(x_1 E_{k-1} y_1 \wedge W_k(y_1,x_2)\big)\big) \\
  &\forall y_1,x_2\;\big(W_k(y_1,x_2) \rightarrow
  \exists y_2\;\big(x_2 E_{k-1} y_2 \wedge Q_{k-1}(y_1,y_2) \wedge \neg\chi(y_1,y_2)\big)\big)\label{eqn:three-variable-trik2}
\end{flalign}

We next postulate the existence of cyclic successors:
\begin{flalign}
  \forall x\;\exists y\;\big(S_k(x,y) \wedge x E_{k} y\big)
\end{flalign}

\smallskip\noindent\emph{Zeros.}
We now decorate with $Z_k$ every $E_{k-1}$-class that represents the $k$-numeral $0$.
First, when $Z_k$ holds, every $B_k$ is off:
\begin{flalign}
  \forall x_1,x_2\;\big(\big(x_1 E_{k-1} x_2 \wedge Z_k(x_1)\big) \rightarrow \neg B_k(x_2)\big)
\end{flalign}
For converse, if $Z_k$ does not hold, then some $B_k$ is on:
\begin{flalign}
  \forall x\;\big(\neg Z_k(x) \rightarrow \exists y\;\big(x E_{k-1} y \wedge B_k(y)\big)\big)
\end{flalign}

At this point, every $E_k$-class decomposes into at least $\tet{k}{n}$ inner $E_{k-1}$-classes.
The $E_{k-1}$-classes representing circularly successive $k$-numerals are linked by $S_k$, provided that they are contained in the same $E_K$-class.
Finally, the $E_{k-1}$-classes representing the $k$-numeral~$0$ are marked by $Z_k$.
It remains therefore to axiomatise the $k$th level equality predicate $Q_k$.

\smallskip\noindent\emph{Equality.}
The predicate $Q_k$ shall connect $E_{k-1}$-classes representing the same $k$-numeral, provided that these classes are contained in the same $E_K$-class.
We use the fact that a common $S_k$-successor exists precisely when the $k$-numerals agree:
\begin{flalign}
  \forall x_1,x_2\;\Big(&x_1 E_K x_2 \rightarrow \exists y\;\big(S_k(x_1,y) \wedge
  \big(Q_k(x_1,x_2) \leftrightarrow S_k(x_2,y)\big)\big)\Big)\label{phik-pre-last}
\end{flalign}

Having the predicate $Q_k$ axiomatised,
we enforce that in every $E_k$-class each $k$-numeral is represented by a unique $E_{k-1}$-class:
\begin{flalign}\label{phik-last}
  \forall x_1,x_2\;\big(\big(x_1 E_k x_2 \wedge Q_k(x_1,x_2)\big) \rightarrow x_1 E_{k-1} x_2\big)
\end{flalign}

Finally, define $\phi_{K,n}^{(k)}$ as the conjunction of \eqref{phik-first}--\eqref{phik-last}.

\subsection{Complexity Lower Bound}\label{sec:GF3-complexity}

Using the sentences $\phi_{K,n}$ constructed in \autoref{sec:GF3-large-counters},
we encode accepting runs of alternating Turing machines operating in $K$-exponential space.
The desired lower bound for the constant-free, equality-free $\GF^3[\KEQ{K}]$ follows then from the result of
Chandra, Kozen, and Stockmeyer~\cite{CKS81}, stating that
$K\text{-}\AExpSpace = (K{+}1)\text{-}\ExpTime$.

\subparagraph*{Alternating machines.}
Let $K \in \N$.
Suppose that $\mathbf{M}$ is an alternating Turing machine that operates in a space bound $\tet{K}{p(n)}$ for some polynomial $p$.

The machine $\mathbf{M}$ is defined by the finite set of \emph{states} $S$ with a distinguished \emph{initial state} $s_0 \in S$;
the finite \emph{tape alphabet} $\Gamma$;
the \emph{transition functions} $\delta_1$ and $\delta_2$;
and a partition of $S$ into \emph{universal}, \emph{existential}, \emph{accepting}, and \emph{rejecting} states.

Let $m = \tet{K}{p(|x|)}$ denote the space bound for the input $x$.
A \emph{configuration} of $\mathbf{M}$ is a complete description of the tape with $m$ cells, the head position, and the current state.
The \emph{run} of $\mathbf{M}$ on input $x$ is represented as a directed tree $T$ with a distinguished root $r$.
Each node $u$ of $T$ is labelled by a configuration $C(u)$, where $C(r)$ is the \emph{initial} configuration.
Every non-leaf node has exactly two children. For each edge $v \to u$, the configuration $C(u)$ is obtained
from $C(v)$ by a single deterministic transition step: the left child applies the transition rule $\delta_1$, and the right child applies $\delta_2$.
W.l.o.g.~only the initial configuration $C(r)$ has state $s_0$ and a sequence of configurations $C(v)$ on any path from the root to a leaf is without repetitions; in particular, by the space bound, the tree $T$ is finite.

Leaves of $T$ carry accepting or rejecting states, while internal nodes carry existential or universal ones.
We define \emph{positivity} inductively: a leaf is positive if its state is accepting;
a node with an existential state is positive if at least one of its children is positive;
a node with a universal state is positive if all its children are positive.
Finally, the machine \emph{accepts} $x$ if the root $r$ is positive.

In the following,
we encode the run of $\mathbf{M}$ on an input string $x$: we construct a $\GF^3[\KEQ{K}]$ sentence $\Theta_{\mathbf{M},x}$
such that $\mathbf{M}$ accepts $x$ iff $\Theta_{\mathbf{M},x}$ is satisfiable.

\subparagraph*{Signature.}
The signature of $\Theta_{\mathbf{M},x}$ consists of the following symbols:
\begin{itemize}
  \item The distinguished symbols $E_1,\dots,E_K$, where every $E_{K}$-class models a single configuration with the inner $E_{K-1}$-classes representing individual tape cells.
  \item The binary symbol $S_K$ and the unary symbol $Z_K$, which are shared with the sentence $\phi_{K,n}$ (where $n \eqdef p(|x|)$) from~\autoref{sec:GF3-large-counters}.
  \item Unary symbols $P_a$ for $a \in \Gamma$ and $P_s$ for $s \in S$.
  \item Unary symbols $P_>$ and $P_<$, marking the head being located to the left and to right, respectively.
  \item Binary symbols $T_i$, $T^L_i$, $T^R_i$ for $i \in \{1,2\}$ for modelling $\delta_i$-transitions, where $T_i$ links the $j$th cell of a configuration to the $j$th cell of a successive configuration, $T^{L}_i$ links to the $(j{-}1)$th cell, and $T^{R}_i$ links to the $(j{+}1)$th cell.
  \item A symbol $F$ to mark positive nodes in the computation tree.
\end{itemize}

\subparagraph*{Configurations.}
Every $E_K$-class shall represent a single configuration of $\mathbf{M}$:
the inner $E_{K-1}$-classes correspond to the individual tape cells;
the sentence $\phi_{K,n}$ from~\autoref{sec:GF3-large-counters} enforces that the correct number of such cells is present and arranges them in a linear order.
The symbol written in each cell is indicated by the unary predicates $P_a$ (for $a \in \Gamma$), which are made $E_{K-1}$-universal.
The current state $s$ and the head position are encoded analogously: exactly one $E_{K-1}$-class is marked with $P_s$, identifying the cell scanned by the head.

A typical configuration $C$ is represented schematically as
\[
\begin{pmatrix}
  a_1 & a_2 & \dots & a_{i-1} & a_i & a_{i+1} & \dots & a_{m-1} & a_m \\
  <   & <   & \dots & <        & s_i & >        & \dots & >        & > 
\end{pmatrix}
\]
Here $a_j \in \Gamma$ denotes the tape symbol in cell $j$, $s_i \in S$ is the machine state, and the markers $<$ and $>$ indicate whether a cell is located to the left or to the right of the head.

Let $\num_K$ be the function assigning $K$-numerals as in \autoref{sec:GF3-large-counters}.
Formally, an $E_K$-class $\cC$ models the configuration $C$ if, for every element $b \in \cC$, the following conditions hold:
\begin{itemize}
  \item for each $a \in \Gamma$, 
    \(\str{A} \models P_a(b) \quad\text{iff}\quad a = a_{\num_K(b)};\)
  \item for each $s \in S$,
    \(\str{A} \models P_{s}(b) \quad\text{iff}\quad \num_K(b) = i \;\text{and}\; s = s_i; \)
  \item for each $\bowtie \, \in \{<,>\}$,
    \(\str{A} \models P_{\bowtie}(b) \quad\text{iff}\quad \num_K(b)\,\bowtie\, i.\)
\end{itemize}

We now express that every $E_K$-class models some configuration.

We first enforce that each element carries some tape symbol and some state marker:
\begin{flalign}
  &\forall x\;\bigvee\nolimits_{a \in \Gamma} P_a(x)
  \quad\wedge\quad
  \forall x\;\bigvee\nolimits_{s \in S \cup \{<,>\}} P_s(x)
\end{flalign}

To guarantee that every $E_{K-1}$-class encodes exactly one tape cell, we express that no two distinct elements of the same $E_{K-1}$-class may carry different tape symbols or different state markers:
\begin{flalign}
  &\forall x_1,x_2\;\Big(x_1 E_{K-1} x_2 \rightarrow \bigwedge\nolimits_{a,a' \in \Gamma \,:\, a \neq a'}
      \bigl(\neg P_a(x_1) \vee \neg P_{a'}(x_2)\bigr)\Big) \\[1mm]
  &\forall x_1,x_2\;\Big(x_1 E_{K-1} x_2 \rightarrow \bigwedge\nolimits_{s,s' \in S \cup \{<,>\} \,:\, s \neq s'}
      \bigl(\neg P_s(x_1) \vee \neg P_{s'}(x_2)\bigr)\Big)
\end{flalign}

Next, we propagate the markers $<$ and $>$ along the successor relation $S_K$ to enforce that exactly one cell carries the state symbol (that is, the cell corresponding to the head position):
\begin{flalign}
  &\forall x_1,x_2\;\Big(\Big(S_K(x_1,x_2) \wedge \neg Z_K(x_2) \wedge \bigvee\nolimits_{s \in S \cup \{>\}} P_s(x_1)\Big) \rightarrow P_>(x_2)\Big) \\[1mm] 
  &\forall x_1,x_2\;\Big(\Big(S_K(x_1,x_2) \wedge \neg Z_K(x_2) \wedge \bigvee\nolimits_{s \in S \cup \{<\}} P_s(x_2)\Big) \rightarrow P_<(x_1)\Big)
\end{flalign}
Above, $Z_K$ is used to detect whether $S_K$ cycles back to zero, or connects two elements successive without cyclicity (i.e., $S_K(x_1,x_2) \wedge \neg Z_K(x_2)$ holds when the cells are consecutive).

Finally, we constrain the boundary cells and rule out adjacent $<$ and $>$ markers.  
That is, the leftmost cell cannot carry $>$; the rightmost cannot carry $<$; and $<$ cannot be immediately followed by $>$ (i.e., there must be the head between them):
\begin{flalign}
  &\forall x_1,x_2\;\Big(\big(S_K(x_1,x_2) \wedge Z_K(x_2)\big) \rightarrow \big(\neg P_<(x_1) \wedge \neg P_>(x_2)\big)\Big) \\[1mm]
  &\forall x_1,x_2\;\Big(\big(S_K(x_1,x_2) \wedge \neg Z_K(x_2)\big) \rightarrow \big(\neg P_<(x_1) \vee \neg P_>(x_2)\big)\Big)
\end{flalign}
Here, $Z_K$ is again used to detect the boundary of the tape. 

\subparagraph*{Transitions.}
Let $\cC$ and $\cC'$ be two $E_K$-classes modelling configurations $C$ and $C'$, respectively.
We axiomatise that whenever a relation $T_i$ links (elements of) $\cC$ to (elements of) $\cC'$, then $C'$ must be the $\delta_i$-successor of $C$.

First, for every tape position $j \in [0,m-1]$,  
there must exist elements $a \in \cC$ and $b \in \cC'$ such that 
$\num_K(a)=j$, $\num_K(b)=j$, and $\str{A} \models T_i(a,b)$.
Analogous constraints are imposed for the shifted relations $T^L_i$ and $T^R_i$:  
for $T^L_i$ we require $\num_K(b)=j-1$, and for $T^R_i$ we require $\num_K(b)=j+1$;
the conditions are modulo $m$.

We enforce these constraints using the following formulas:
\begin{flalign}\label{phi-trans-first}
  \forall x_1,x_2\;&\big(T_i(x_1,x_2) \rightarrow \big(Z_K(x_1) \leftrightarrow Z_K(x_2)\big)\big) \\
  \forall x_1,x_2\;&\big(T_i(x_1,x_2) \rightarrow \exists y\;\big(S_K(x_1,y) \wedge x_1 E_K y \wedge T^L_i(y,x_2)\big)\big) \\
  \forall x_1,x_2\;&\big(T_i(x_1,x_2) \rightarrow \exists y\;\big(S_K(x_2,y) \wedge x_2 E_K y \wedge T^R_i(x_1,y)\big)\big) \\
  \forall x_1,x_2\;&\big(T^L_i(x_1,x_2) \rightarrow \exists y\;\big(S_K(x_2,y) \wedge x_2 E_K y \wedge T_i(x_1,y)\big)\big) \\
  \forall x_1,x_2\;&\big(T^R_i(x_1,x_2) \rightarrow \exists y\;\big(S_K(x_1,y) \wedge x_1 E_K y \wedge T_i(y,x_2)\big)\big)\label{phi-trans-last}
\end{flalign}

We require the existence of $\delta_i$-successive configurations:
\begin{flalign}
  \bigwedge\nolimits_{i \in \{1,2\}}\forall x\;\Big(\Big(\bigvee\nolimits_{s \in S \,:\, s\text{ is universal or existential}} P_s(x)\Big) \rightarrow \exists y\;T_i(x,y)\Big)
\end{flalign}

Having the above formulas, it is routine to express transition functions of $\mathbf{M}$.
For instance, to express cell content staying the same for positions not under the head, we write:
\begin{flalign}
  \forall x_1,x_2\;\Big(&\Big(T_i(x_1,x_2) \wedge \bigwedge\nolimits_{s \in S} \neg P_s(x_1)\Big) \rightarrow \bigwedge\nolimits_{a \in \Gamma}\big(P_a(x_1) \leftrightarrow P_a(x_2) \big)\Big)
\end{flalign}

\subparagraph*{Accepting condition.}
To express the accepting condition, we decorate $E_K$-classes modelling positive nodes in the computation tree with the unary predicate $F$.

First, require the symbol $F$ to be $E_K$-universal:
\begin{flalign}
  \forall x_1,x_2\; \big(x_1 E_K x_2 \rightarrow \big(F(x_1) \leftrightarrow F(x_2)\big)\big)
\end{flalign}
Then propagate the positivity condition through the computation tree:
\begin{flalign}
  &\forall x\; \big(P_{s_0}(x) \rightarrow F(x)\big) \quad \wedge \quad \forall x\; \Big(\Big(\bigvee\nolimits_{s \in S \,:\, s\text{ is rejecting}} P_{s}(x)\Big) \rightarrow \neg F(x)\Big) \\
  &\forall x\; \Big(\Big(F(x) \wedge \bigvee\nolimits_{s \in S \,:\, s\text{ is universal}} P_s(x)\Big) \rightarrow \bigwedge\nolimits_{i \in \{1,2\}}\exists y \;\big(T_i(x,y) \wedge F(y)\big)\Big) \label{eqn:forall-positivity} \\
  &\forall x\; \Big(\Big(F(x) \wedge \bigvee\nolimits_{s \in S \,:\, s\text{ is existential}} P_s(x)\Big) \rightarrow \bigvee\nolimits_{i \in \{1,2\}}\exists y \;\big(T_i(x,y) \wedge F(y)\big)\Big) \label{eqn:exists-positivity}
\end{flalign}

To conclude the reduction, it remains to express the initial condition: the tape content of a configuration with the initial state $s_0$ shall contain the input string $x$.
Yet, this final step is relatively straightforward, so we leave details to the reader.

\subsection{Lifting Lower Bound}\label{sec:lifting-lower-bounds}

We now extend the lower-bound construction to
obtain the $(K{+}2)$-\ExpTime-hardness of the satisfiability problem for $\GF[\KEQ{K}]$ (with constants and unboundedly many variables).

Recall that in \autoref{sec:GF3-large-counters}, the formulas $\phi_{S}$, $\phi_{Z}$, and $\phi_{Q}$ were used to encode numbers in the range $[0,2^n-1]$, relying on $n$ unary predicates $BB_i$ to represent individual bits. 

The key insight underlying the $(K{+}2)$-\ExpTime-hardness of $\GF[\KEQ{K}]$ is that numbers from the much larger range $[0,2^{2^n}-1]$ can already be encoded at the level of individual elements: 
let $BB$ be a predicate of arity $n+1$, and let $0$ and $1$ be constants identified with the usual numbers $0$ and $1$.
An element $a \in A$ represents a number $m \in [0,2^{2^n}-1]$ if, for every sequence
$j_0,\dots,j_{n-1} \in \{0,1\}$, we have
$\str{A} \models BB(a,j_{n-1},\dots,j_{0})$ precisely when the $k$-th bit of
$m$ is equal to~$1$, where
$k = \sum_{i=0}^{n-1} j_i \cdot 2^{i}$.
Clearly, this condition cannot be expressed directly for all
sequences $j_0,\dots,j_{n-1}$, as a central requirement is
that the resulting sentences shall have length polynomial in~$n$.
We therefore employ a succinct encoding technique, whose details are deferred to \autoref{appendix:GF3-lower}.
A closely related construction appears in the proof of $2$-\NExpTime{}-completeness for the Triguarded Fragment~\cite{RS18}.

Now, let $\phi_{K,n} = \phi^{(1)}_{K,n} \wedge \dots \wedge \phi^{(K)}_{K,n}$ be as constructed in
\autoref{sec:GF3-large-counters}.
We replace the sentences~\eqref{unary-arithmetic-0}--\eqref{unary-arithmetic-2}
occurring in $\phi^{(1)}_{K,n}$ with their succinct counterparts, yielding an
exponential increase in the representable counting range.
All remaining components of the lower-bound construction from
Subsections~\ref{sec:GF3-large-counters} and~\ref{sec:GF3-complexity} remain
unchanged.
As a result, we can encode runs of alternating Turing machines operating in
$(K{+}1)$-exponential space, rather than $K$-exponential space.
This yields $(K{+}2)$-\ExpTime-hardness for the
equality-free $\GF[\KEQ{K}]$.

\section{Upper Bound}\label{sec:GF-upper}

Having established the lower bounds, we now turn to upper bounds (Theorems~\ref{thm:GF-fmp} and~\ref{thm:GFmain}).
To avoid technical distraction unrelated to nested equivalence relations,
the results in this section are established for the constant-free, equality-free $\GF^3[\EQ]$.
In Appendix~\ref{appendix:GF-upper}, we provide full proofs for the general case,
that is, for the equality-free $\GF[\EQ]$ with constants and an unbounded number of variables.

\subsection{Normal Form and Atomic Types}

\begin{assumption}\label{assumption:GF3-normal-form}
  For simplicity, in this section we work with the constant-free, equality-free $\GF^3[\EQ]$.
  Given a sentence $\phi$ of this fragment, we assume that it is in \emph{normal form}:
  \begin{flalign*}
      \Psi
      \quad\wedge\quad 
      \bigwedge\nolimits_{t \in \cI}\forall x_1,x_2\, \big(G_t(x_1,x_2) \rightarrow \exists y\, \psi_t(x_1,x_2,y)\big), 
  \end{flalign*}
  where
  each $G_t$ is a binary guard,
  each $\psi_t$ is quantifier-free, and
  $\Psi$ is a conjunction of purely universal prenex sentences:
  \begin{flalign*} 
    \bigwedge\nolimits_{t \in \cJ}\forall x_1,\dots,x_{k_t}\, \big(F_t(x_1,\dots,x_{k_t}) \rightarrow \chi_t(x_1,\dots,x_{k_t})\big), 
  \end{flalign*}
  where $1 \le k_t \le 3$, each $F_t$ is a guard of arity $k_t$, and each $\chi_t$ is quantifier-free. 

  The guards $G_t$ and $F_t$ are allowed to be distinguished symbols.
\end{assumption}

Using standard techniques (see, e.g.,~\cite{Gra99,PH23}),
given an arbitrary $\GF^3$-sentence, one can compute in polynomial time a sentence in normal form which is satisfiable over the same domains.
This reduction remains sound also for the fragment $\GF^3[\EQ]$.
Hence, for establishing decidability and the finite model property for $\GF^3[\EQ]$, Assumption~\ref{assumption:GF3-normal-form} is w.l.o.g.

\subparagraph*{Types.}
Let $\sigma$ be a constant-free signature.
For $k \in \N$, let $\operatorname{Lit}_k(\sigma)$ denote the set of all literals over relation symbols from $\sigma$ and variables $X_k \eqdef \{ x_1,\dots,x_k \}$; literals with ``='' are not included.
A \emph{$k$-type} over $\sigma$ is a maximal consistent subset of $\operatorname{Lit}_k(\sigma)$,
that is, for every atom $\gamma$ over $\sigma$ and $X_k$, precisely one of $\gamma$, $\neg\gamma$ is present.
A \emph{type} is a $k$-type for some $k \ge 1$.  

Given a $\sigma$-structure $\fA$ and $k$ distinct elements $a_1,\dots,a_k \in A$, we write 
\( \type{\fA}{a_1,\dots,a_k} \) 
for the $k$-type \emph{realised} by $\langle a_1,\dots,a_k \rangle$ in $\fA$.  
Formally, $\type{\fA}{a_1,\dots,a_k}$ is the set of literals $\gamma \in \operatorname{Lit}_k(\sigma)$ such that $\fA,f \models \gamma$, where $f$ sends $x_i \mapsto a_i$ for all $i \in [k]$.

It is convenient to view $k$-types themselves as $\sigma$-structures over the canonical domain $\{x_1,\dots,x_k\}$.
We allow structure-like notation also for types; for instance, if $\tau$ is a $3$-type, we may write \( \type{\tau}{x_1,x_3} \) for the $2$-type induced by $\tau$ on the variables $x_1$ and $x_3$.

\subsection{Finite-Index Nesting Property}\label{sec:GF3-fin-property}

Our main technical tool is the \emph{finite-index nesting property}:
namely, if a sentence $\phi$ is satisfiable, then it admits a model in which the finest distinguished relation $E_1$ sits inside $E_2$ with \emph{finite index}, that is, every $E_2$-class decomposes into only finitely many $E_1$-classes.

Lemma~\ref{l:GF3-fin-property} establishes that the finite-index nesting property holds for the constant-free, equality-free $\GF^3[\EQ]$.
It also provides an upper bound on the \emph{index} of $E_1$ in every $E_2$-class, that is, on the number of inner $E_1$-classes sitting in an $E_2$-class.

\begin{lemma}\label{l:GF3-fin-property}
  Let $K \ge 2$,
  and let $\phi$ be a constant-free, equality-free sentence of $\GF^3[\KEQ{K}]$ in normal form as in Assumption~\ref{assumption:GF3-normal-form}.
  Let $\AAA$ denote the set of all $1$-types over the signature of~$\phi$.
  If $\phi$ is satisfiable, then it has a model $\fA \models \phi$ such that
  \[
    |\cC / E^{\fA}_1| \;\le\; 3 \cdot 2^{|\AAA|}
    \qquad\text{for every}\qquad
    \cC \in A / E^{\fA}_2.
  \]
\end{lemma}

The key idea is that
Lemma~\ref{l:GF3-fin-property} allows us to eliminate
the first distinguished symbol $E_1$ from the logic:
a finite number of equivalence classes can be fully described without relying on the externally constrained semantics of the predicate $E_1$.
For instance, we can introduce a family of fresh unary predicates $\{ U_i \}_{i=1}^{n}$, for an appropriate $n \in \N$, and
decorate elements of the domain so that the following equivalence holds:
\[ x E_1 y \; \leftrightarrow \; \Big(x E_2 y \wedge \bigwedge\nolimits_{i=1}^{n}\big(U_i(x) \leftrightarrow U_i(y)\big)\Big). \]
This transformation of models serves as a basis for a reduction from $\GF^3[\KEQ{K}]$ to the simpler fragment $\GF^3[\KEQ{(K{-}1)}]$,
and by iterative applications to $\GF^3[\KEQ{1}]$. 
Technical details of this reduction are developed in the subsequent \autoref{sec:GF3-reduction},
and we now focus on proving Lemma~\ref{l:GF3-fin-property}.

\subparagraph*{Proof of Lemma~\ref{l:GF3-fin-property}.}
Let $K \ge 2$, and let $\phi$ be a constant-free, equality-free sentence of $\GF^3[\KEQ{K}]$
in normal form as in Assumption~\ref{assumption:GF3-normal-form}.
We write $\sigma$ for the signature of $\phi$, and $\AAA$ for the set of all $1$-types over $\sigma$.
The signature $\sigma$ includes the first $K$ distinguished symbols $E_1,\dots,E_K$ and some number of non-distinguished ones.

We begin with a simple lemma about equality-free \FO{}.

\begin{lemma}\label{l:propermodels}
  If $\phi$ is satisfiable,
  then there exists a model~$\str{A} \models \phi$ with the following properties.
  \begin{enumerate}
    \item\label{l:propermodelsitem1} For every index $t \in \cI$ and every pair $(a_1,a_2) \in A^2$ with $\str{A} \models G_t(a_1,a_2)$,
    there exists $a_3 \in A \setminus \{ a_1,a_2\}$ such that
    \[ \str{A} \models \psi_t(a_1,a_2,a_3). \]
    \item\label{l:propermodelsitem2} If $\cC \in A/E^{\str{A}}_1$ and $\alpha \in \AAA$, then the set
    \[ \{ a \in \cC \mid \type{\str{A}}{a} = \alpha \} \]
  is either empty or countably infinite.
  \end{enumerate}
\end{lemma}
\begin{proof}
  Let $\fA_0$ be a model of $\phi$.
  We define the structure $\fA$ over the domain $A_0 \times \N$:
  for each relation symbol $R \in \sigma$, denoting its arity as $k$, interpret $R$ in ${\str{A}}$ as the relation
  \[ \bigcup_{i_1,\dots,i_k \in \N}\big\{\, ((a_1,i_1),\dots(a_k,i_k)) \bigmid (a_1,\dots,a_k) \in R^{\str{A}_0} \,\big\}.\]
  The required properties are readily verified.
\end{proof}

To the end of this subsection,
assuming that $\phi$ is satisfiable, we fix a model $\str{A} \models \phi$ as in Lemma~\ref{l:propermodels},
and subsequently, construct a new model $\str{B} \models \phi$ that will satisfy Lemma~\ref{l:GF3-fin-property}.

\subparagraph*{Sorts and spectra.}
We introduce the following definitions.
\begin{itemize}
  \item A \emph{class-sort} of an element $a \in A$ is the set of $1$-types realised by elements that are $E^{\str{A}}_1$-related to $a$; formally:
\[ \classsort{\str{A}}{a} \eqdef \big\{ \type{\str{A}}{b} \bigmid b \in \absclass{\str{A}}{1}{a} \big\} \subseteq \AAA. \]

  We also define $\classsort{\str{A}}{\cC} \eqdef \classsort{\str{A}}{a}$, where
  $\cC \in A/E^{\str{A}}_1$ and $a \in \cC$ is an arbitrarily chosen representative.
  For $\Gamma \subseteq \AAA$, the $E_1$-class $\cC$ is said to \emph{realise} $\Gamma$ if $\classsort{\str{A}}{\cC} = \Gamma$.

  \item Let $S \subseteq A$ be such that $S/E^{\str{A}}_1 \subseteq A/E^{\str{A}}_1$ (i.e., $S$ is the union of some $E_1$-classes).
  Then a \emph{class-spectrum} of $S$ is a function $2^{\AAA} \rightarrow \N$, denoted as $\classspectrum{\str{A}}{S}$, that sends $\Gamma \subseteq \AAA$ to
\[ \classspectrum{\str{A}}{S}(\Gamma) \eqdef \big|\big\{ \cC \in S/E^{\str{A}}_1 \bigmid \classsort{\str{A}}{\cC} = \Gamma \big\}\big|. \]
\end{itemize}

\subparagraph*{Truncation.}
We now select a subset of the domain $A$ that will serve as a foundation for the new model.
It will contain a \emph{small} number of $E_1$-classes inside every $E_2$-class but at the same time will be rich enough to provide all necessary witnesses.
To define this subset, we employ an auxiliary operation which we call a \emph{$3$-truncation}.

Let $\cC \in A/E^{\str{A}}_{2}$.
A subset $T \subseteq \cC$ is a \emph{$3$-truncation} of $\cC$ if it satisfies the following conditions:
\begin{romanenumerate}
  \item $T/E^{\str{A}}_1 \subseteq \cC/E^{\str{A}}_1$;
  \item for every $\Gamma \subseteq \AAA$, we have
  \[ \classspectrum{\str{A}}{T}(\Gamma) \;=\; \inf\big\{\, \classspectrum{\str{A}}{\cC}(\Gamma),\, 3 \, \big\}. \]
\end{romanenumerate}
In other words,
$T$ is the union of inner $E_1$-classes of $\cC$ which considers at most $3$ representatives of every class-sort $\Gamma \subseteq \AAA$.
The parameter $3$ is chosen because $\phi$ uses $3$ variables.

For every $\cC \in A/E^{\str{A}}_2$, we arbitrarily fix a $3$-truncation for $\cC$ and denote it as $T_3(\cC)$.
The existence of $T_3(\cC)$ is straightforward, but naturally its choice might not be unique.

Finally, we define the desired subset as
\[B \eqdef \bigcup_{\cC \in A/E^{\str{A}}_2}T_3(\cC).\]

\subparagraph*{Similar pairs.}
We now introduce a \emph{similarity} relation $\sim$ on pairs of elements that exhibit analogous characteristics when considering their $1$-types and interactions through the nested equivalence relations $E^{\str{A}}_1 \subseteq \dots \subseteq E^{\str{A}}_K$.

Let $(a_1,a_2) \in A^2$ and $(b_1,b_2) \in A^2$. 
We write $(a_1,a_2) \sim (b_1,b_2)$ if the following holds:
\begin{romanenumerate}
  \item \label{patternitem0} $a_1 = a_2$ if and only if $b_1 = b_2$.
  \item \label{patternitem1} For $i \in \{ 1, 2\}$, $\type{\str{A}}{a_i} = \type{\str{A}}{b_i}$.
  \item \label{patternitem2} For $i \in \{1, 2 \}$, $\classsort{\str{A}}{a_i} = \classsort{\str{A}}{b_i}$.
  \item $(a_1,a_2) \in E^{\str{A}}_1$ if and only if $(b_1,b_2) \in E^{\str{A}}_1$.
  \item \label{patternitem3} For every $k \in [2,K]$, $(a_1,b_1) \in E^{\str{A}}_k$ and $(a_2,b_2) \in E^{\str{A}}_k$.
\end{romanenumerate}

It is readily verified that $\sim$ is indeed an equivalence relation.
The following lemma will be used in the construction of a new model for $\phi$.

\begin{lemma}\label{l:patterns}
  Let $(a_1,a_2) \in A^2$ and $(b_1,b_2) \in B^2$. 
  Suppose that
  \[ \text{$(a_1,a_2) \sim (b_1,b_2)$ and $F \subseteq B$ is a finite set with $\{ b_1,b_2 \} \subseteq F$.}\]
  Then, for any $a_3 \in A \setminus \{ a_1,a_2 \}$,
  one can find an element $b_3 \in B \setminus F$ satisfying that
  \[ (a_1,a_3) \sim (b_1,b_3) \quad \text{and} \quad (a_2,a_3) \sim (b_2,b_3). \]
\end{lemma}
\begin{proof}
  We consider two cases.
  
    First, suppose that either $(a_1,a_3) \in E^{\str{A}}_1$ or $(a_2,a_3) \in E^{\str{A}}_1$;
    w.l.o.g. assume the former.
    We arbitrarily select $b_3$ from
    \[ \big\{ b \in \absclass{\str{A}}{1}{b_1} \setminus F \bigmid \type{\str{A}}{b} = \type{\str{A}}{a_3} \big\}. \]
    The above set is non-empty, since $F$ is assumed to be finite and yet the equivalence class $\absclass{\str{A}}{1}{b_1}$ realises every $1$-type of $\classsort{\str{A}}{b_1}$ ($= \classsort{\str{A}}{a_1}$) infinitely often (item~\ref{l:propermodelsitem2} of Lemma~\ref{l:propermodels}).

    Now, assume that neither $(a_1,a_3)$ nor $(a_2,a_3)$ is $E^{\str{A}}_1$-related.
    Let $\cC = \absclass{\str{A}}{2}{a_3}$, and let $\cD = \absclass{\str{A}}{1}{a_3}$.
    We claim that one can select an $E_1$-class $\cD' \in T_3(\cC)/E^{\str{A}}_1$
    such that
    \[ \cD' \cap \{ b_1,b_2 \} = \emptyset \quad\text{and}\quad \classsort{\str{A}}{\cD'} = \classsort{\str{A}}{\cD}. \]
    The existence of such $E_1$-class $\cD'$ follows from the definition of $3$-truncation:
    $T_3(\cC)$ includes the same number of realisations of $\classsort{\str{A}}{\cD}$ as $\cC$ does when this number is at most $2$;
    and otherwise it includes precisely $3$ realisations of $\classsort{\str{A}}{\cD}$.
    In either case, this suffices to avoid---when necessary---the two ``bad'' $E_1$-classes: namely, $\absclass{\str{A}}{1}{b_1}$ and $\absclass{\str{A}}{1}{b_2}$.

    Finally, we arbitrarily select $b_3$ from
    \[ \big\{ b \in \cD' \setminus F \bigmid \type{\str{A}}{b} = \type{\str{A}}{a_3} \big\}. \]
    As before this set is non-empty.
\end{proof}

\subparagraph*{Model construction.}
We now construct a chain of finite structures
\[ \str{B}_1 \subseteq \str{B}_2 \subseteq \str{B}_3 \subseteq \dots, \quad\quad \text{where } \str{B}_{n+1} \,\restr\, B_n = \str{B}_n \text{ for all } n \ge 1. \]
The natural limit structure $\str{B}_{\omega} = \bigcup_{n \ge 1}{\str{B}_n}$ will be a model of $\phi$ satisfying Lemma~\ref{l:GF3-fin-property}.
For every $n \in \N$, the structure $\str{B}_n$ will be defined over a subset $B_n \subseteq B$,
and the inclusion mapping $\str{B}_n \hookrightarrow \str{A}$ will preserve the $1$-types and the nested equivalence relations $E_1,\dots,E_K$; though not necessarily other atoms.
Formally, this requirement expresses as:
\begin{itemize}
  \item For every $b \in B_n$, $\type{\str{B}_n}{b} = \type{\str{A}}{b}$.
  \item For every $k \in [K]$, $E^{\str{B}_n}_k = E^{\str{A}}_k \cap B^2_n$.
\end{itemize}

To guide the construction, we will additionally define a sequence of mappings $(\rho_n)_{n \ge 1}$:
\begin{enumerate}
  \item[$(\dagger)$]\label{invariant-of-rho} Let $n \ge 1$. If $(b_1,b_2) \in B^2_n$ is such that $\str{B}_n \models G_t(b_1,b_2)$ for some index $t \in \cI$,
  then the mapping $\rho_n$ sends $(b_1,b_2)$ to a pair $(a_1,a_2) \in A^2$ such that $(b_1,b_2) \sim (a_1,a_2)$ and
  \begin{itemize}
    \item $\type{\str{B}_n}{b_1} = \type{\str{A}}{a_1}$ when $b_1 = b_2$;
    \item $\type{\str{B}_n}{b_1,b_2} = \type{\str{A}}{a_1,a_2}$ when $b_1 \neq b_2$.
  \end{itemize}
\end{enumerate}

We now formulate a sufficent criterion for satisfaction of the sentence $\phi$ in $\fB_{\omega}$.

A pair of elements $(b_1,b_2) \in B_n^2$ is \emph{satisfied} in $\str{B}_n$ if $\str{B}_n \models G_t(b_1,b_2) \rightarrow \exists y\;\psi_t(b_1,b_2,y)$ for every $t \in \cI$, and otherwise it is \emph{unsatisfied}.
It is readily verified that $\fB_{\omega} \models \phi$ whenever:
\begin{enumerate}
  \item For every $n \ge 1$, the structure $\str{B}_n$ interprets $E_1,\dots,E_K$ as nested equivalence relations.
  \item For every $n \ge 1$, we have $\str{B}_n \models \Psi$.
  \item Every pair $(b_1,b_2) \in B_n^2$ with $n \geq 1$ is eventually satisfied in some structure $\str{B}_\ell$ with $\ell \geq n$;
although no finite global upper bound on such~$\ell$ is required for all pairs simultaneously.
\end{enumerate}

\smallskip\noindent\emph{Base case.}
Let $\str{B}_1 \eqdef \str{A}\,\restr\,\{ a \}$, where $a \in B$ is an arbitrarily chosen element. 
The corresponding mapping is $\rho_1 \eqdef \{ (a,a) \mapsto (a,a) \}$.

\smallskip\noindent\emph{Inductive step.}
Let $n \ge 1$.
Assume that $\str{B}_n$ and $\rho_n$ have already been defined.
If there are no unsatisfied pairs in $\str{B}_n$, we conclude our construction and declare $\str{B}_{\omega} \eqdef \str{B}_n$ as the final model.
Yet, suppose that an unsatisfied pair exists:
in the current step choose one such pair, denote it as $(b_1,b_2) \in B^2_n$.
Let $t \in \cI$ be an index for which the pair $(b_1,b_2)$ lacks a witness, that is,
such that $\str{B}_n \models G_t(b_1,b_2)$ holds, but $\str{B}_n \models \neg \exists y\;\psi_t(b_1,b_2,y)$.

We require that pairs are chosen in a \emph{fair way}, which means that no pairs remain omitted for infinitely many steps. It is clear that such a strategy for selecting unsatisfied pairs exists.

Now, we show how to construct the successive structure $\str{B}_{n+1}$ extending $\str{B}_n$ by one fresh element $b_3 \in B \setminus B_n$ so that $\str{B}_{n+1} \models \psi_t(b_1,b_2,b_3)$ and inductive invariants are preserved. 
We consider only the case $b_1 \neq b_2$, as the case $b_1 = b_2$ follows a similar line of arguments.

Let $(a_1,a_2) \eqdef \rho_n(b_1,b_2)$.
By the inductive invariant of $\rho_n$, we have that $(a_1,a_2) \sim (b_1,b_2)$ and $\type{\fA}{a_1,a_2} = \type{\str{B}_n}{b_1,b_2}$.
Fix a witness $a_3 \in A$ for the pair $(a_1,a_2)$, that is, an element satisfying $\str{A} \models \psi_t(a_1,a_2,a_3)$.
By item~\ref{l:propermodelsitem1} of Lemma~\ref{l:propermodels}, the witness $a_3$ can be chosen to be distinct from $a_1$ and $a_2$.
Applying Lemma~\ref{l:patterns}, with $F \eqdef B_n$, we obtain $b_3 \in B \setminus B_n$ satisfying
$(a_1,a_3) \sim (b_1,b_3)$ and $(a_2,a_3) \sim (b_2,b_3)$.

We now construct the structure $\str{B}_{n+1}$ by extending $\str{B}_{n+1}$ with the element $b_3$.
\begin{itemize}
  \item Initialise $\str{B}_{n+1} \eqdef \str{B}_{n}$.
  \item The witnessing $3$-type: assign $\type{\str{B}_{n+1}}{b_1,b_2,b_3} \eqdef \type{\str{A}}{a_1,a_2,a_3}$.
  \item $2$-types induced by the distinguished equivalence relations:
  for every $b \in B_n \setminus \{ b_1,b_2\}$
  such that $(b,b_3) \in E^{\str{A}}_K$, assign $\type{\str{B}_{n+1}}{b,b_3} \eqdef \type{\str{A}}{b,b_3}$.
  \item Remaining facts whose truth values has been not specified are set to false.
\end{itemize}

The construction of $\str{B}_{n+1}$ is now finished.
Observe that $\str{B}_{n+1} \,\restr\, B_n = \str{B}_n$,
and also that $\str{B}_{n+1}$ interprets $E_1,\dots,E_K$ in the same way as $\str{A} \,\restr\, B_{n+1}$ does.

We now argue that $\str{B}_{n+1} \models \Psi$.
Recall that $\Psi$ is a conjunction of purely universal prenex sentences.
Consider one of its conjuncts: \(\forall x_1,\dots,x_k \, (F(x_1,\dots,x_k) \to \chi(x_1,\dots,x_k)),\)
where $1 \leq k \leq 3$, $F$ is a guard, and $\chi$ is quantifier-free.
Take any tuple $b_1,\dots,b_k \in B_{n+1}$ such that \(\str{B}_{n+1} \models F(b_1,\dots,b_k).\)
Since $\chi$ is quantifier-free, the truth of $\chi(b_1,\dots,b_k)$ depends only on the atomic type induced by the chosen elements.
Thus, the requirement $\str{B}_{n+1} \models \Psi$ reduces to showing the following: every \emph{guarded} type realised in $\str{B}_{n+1}$ is already realised in the original model~$\str{A}$,
where a $k$-type $\tau$ is said to be guarded if it contains a positive literal $\gamma$ with
\(\FreeVars(\gamma)=\{x_1,\dots,x_k\},\) that is, if some positive literal uses all variables of~$\tau$.

The structure $\str{B}_{n+1}$ is obtained from $\str{B}_n$ by copying atomic types from~$\str{A}$ while preserving all previously assigned types.
This construction cannot introduce new guarded types---in contrast to non-guarded ones.
It follows that $\str{B}_{n+1} \models \Psi$.

We now extend the mapping $\rho_n$.
For brevity, we write $G(x_1,x_2) \eqdef \bigvee_{t \in \cI} G_t(x_1,x_2)$.
\begin{itemize}
  \item Initialise $\rho_{n+1} \eqdef \rho_n$.
  \item For every $(i,j) \in \{1,2,3\}^2 \setminus \{ 1,2\}^2$,
  if $\str{B}_{n+1} \models G(b_i,b_j)$, then put $\rho_{n+1}(b_i,b_j) \eqdef (a_i,a_j)$.
  \item For every element $b \in B_n \setminus \{ b_1,b_2\}$,
  if $\str{B}_{n+1}\models G(b,b_3)$, then put $\rho_{n+1}(b,b_3) \eqdef (b,b_3)$;
  and do symmetrically for $(b_3,b)$: if $\str{B}_{n+1}\models G(b_3,b)$, then put $\rho_{n+1}(b_3,b) \eqdef (b_3,b)$.
\end{itemize}

This concludes the inductive step. The following claim highlights the key points:

\begin{claim}\label{constructioninvariantclaim}
  \, 
  \begin{enumerate}
    \item $\str{B}_{n+1}~\restr~B_n = \str{B}_n$ and $B_{n+1} \setminus B_n = \{ b_3 \}$.
    \item For every $b \in B_{n+1}$, $\type{\str{B}_{n+1}}{b} = \type{\str{A}}{b}$.
    \item For every $k \in [K]$, $E^{\str{B}_{n+1}}_k = E^{\str{A}}_k \cap B^2_{n+1}$.
    \item The invariant~$(\dagger)$ of the mapping $\rho_{n+1}$ is preserved.
    \item $\str{B}_{n+1} \models \Psi$.
    \item $\fB_{n+1} \models \psi_t(b_1,b_2,b_3)$.
  \end{enumerate}
\end{claim}

\subparagraph*{Bounding index.}
To finish the proof of Lemma~\ref{l:GF3-fin-property},
it remains to bound the number of $E_1$-classes in every $E_2$-class of $\str{B}_{\omega}$.
Since $B_{\omega} \subseteq B$ and $\str{B}_{\omega} \hookrightarrow \str{A}$ preserves both $E_1$ and $E_2$, we may alternatively bound the number of $E_1$-classes in every $E_2$-class of $\str{A} \,\restr\, B$.

By definition of $B$, every $\cD \in B/E^{\str{A}}_2$ is a $3$-truncation of some $\cC \in A/E^{\str{A}}_2$, i.e., $\cD = T_3(\cC)$.
The number of distinct class-sorts---i.e., subsets of $\AAA$---is $2^{|\AAA|}$.
Since at most $3$ realisations of every class-sort $\Gamma \subseteq \AAA$ are included in $T_3(\cC)$,
we obtain that $|T_3(\cC)/E^{\str{A}}_1| \le 3 \cdot  2^{|\AAA|}$.

\subsection{Elimination of Distinguished Symbols}\label{sec:GF3-reduction}

We now show a reduction from $\GF^3[\KEQ{K}]$ (where $K \ge 2$) to $\GF^3[\KEQ{1}]$.

\begin{lemma}\label{l:GF3-reduction}
  Let $K \ge 2$,
  and let $\phi$ be a constant-free, equality-free sentence of $\GF^3[\KEQ{K}]$.
  Then there exists a constant-free, equality-free sentence $\phi'$ of $\GF^3[\KEQ{1}]$ such that:
  \begin{itemize}
    \item The sentences $\phi$ and $\phi'$ are equisatisfiable.
    \item If $\phi'$ has a model over a domain $D$, then $\phi$ has also.
    \item $|\phi'|$ is bounded by a $(K{-}1)$-exponential function of $|\phi|$.
    \item $\phi'$ is computable from $\phi$ in $(K{-}1)$-exponential time.
  \end{itemize}
\end{lemma}

\subparagraph*{Proof of Lemma~\ref{l:GF3-reduction}.}
Let $K \ge 2$,
and let $\phi \in \GF^3[\KEQ{K}]$ be a constant-free, equality-free sentence; w.l.o.g $\phi$ is in normal form as in Assumption~\ref{assumption:GF3-normal-form}.
Write $\sigma$ for the signature of $\phi$,
and $\AAA$ for the set of all $1$-types over~$\sigma$.

Set $n = \lceil \log_2 N \rceil$,
where $N = 3 \cdot 2^{|\AAA|}$ is the bound of Lemma~\ref{l:GF3-fin-property}.
We axiomatise the finest distinguished symbol $E_1$ using fresh unary predicates $\big\{ U_i \big\}^n_{i=1}$:
we declare two elements $E_1$-equivalent if and only if they are $E_2$-related and agree on $U_i$ for all $i \in [n]$.

Let $\phi'$ denote the sentence $\phi$ with conjoined axioms:
\begin{flalign}\label{naive-reduction-axiom1}
  \forall x_1,x_2\;\Big(&x_1 E_1 x_2 \rightarrow x_1 E_2 x_2\Big) \\
  \forall x_1,x_2\;\Big(&x_1 E_2 x_2 \rightarrow \Big(x_1 E_1 x_2 \leftrightarrow \bigwedge\nolimits^{n}_{i=1} \big(U_i(x_1) \leftrightarrow U_i(x_2)\big)\Big)\Big)\label{naive-reduction-axiom2}
\end{flalign}

In $\phi'$ the symbol $E_1$ is treated as a non-distinguished binary symbol, i.e., its interpretation is not constrained.
Hence, it can be naturally viewed as a sentence of $\GF^3[\KEQ{(K{-}1)}]$, with the distinguished symbols being $E_2,\dots,E_K$;
this numeration of distinguished symbols does not follow our convention, but it is convenient in the following claim. 

\begin{claim}\label{claim:GF3-reduction-correct}
  \,
  \begin{enumerate}
    \item The sentences $\phi$ and $\phi'$ are equisatisfiable.
    \item The $\sigma$-reduct of any model $\str{A}\models\phi'$ satisfies $\phi$.
    \item $|\phi'|$ is bounded by an exponential function of $|\phi|$.
  \end{enumerate}
\end{claim}
\begin{proof}
It is immediate that the $\sigma$-reduct of any model for $\phi'$ is also a model for $\phi$,
and hence satisfiability of $\phi'$ implies that of $\phi$.

The converse direction is more delicate:
let $n$ and $N$ be as above.
Let $\str{A}\models\phi$ be a model satisfying Lemma~\ref{l:GF3-fin-property}.
Since there are $2^n \ge N$ possible $1$-types over $\{ U_i \}^{n}_{i=1}$,
we have enough combinations to decorate the elements of $\str{A}$ with the predicates $U_i$
so that any $E^{\str{A}}_2$-related elements get the same combination precisely when they are also related by $E^{\str{A}}_1$.

Finally, since $\sigma$ is constant-free, the number of $1$-types is singly exponential in $|\phi|$; more precisely: $|\AAA| = 2^{|\sigma|} \le 2^{|\phi|}$.
We thus have $|\phi'| = |\phi| + \cO(n) = |\phi| + \cO(|\AAA|) = 2^{\cO(|\phi|)}$.
\end{proof}

Having Claim~\ref{claim:GF3-reduction-correct} established, we adjust the numeration of the distinguished symbols:
that is, we rename $E_1$ to a fresh symbol and shift $E_{k+1} \mapsto E_k$ for all $k \ge 1$.
In this way, we obtain a proper sentence of $\GF^3[\KEQ{(K{-}1)}]$.
By recursively applying the above reduction $K-1$ times, 
we reach the final sentence in $\GF^3[\KEQ{1}]$ with the required properties.

\subsection{Finite Model Property}\label{sec:GF3-fmp}

We establish the finite model property via a reduction to an auxiliary fragment denoted \GFU{}.
This fragment is just \GF{} extended with a single distinguished binary symbol~$U$.
Models of \GFU{}-sentences are constrained to interpret $U$ as the universal relation. 
(In expressive power, \GFU{} coincides with the Triguarded Fragment, see~\cite{RS18} for further details.)

To obtain the reduction to $\GFU^3$,
we first embed $\GF^3[\KEQ{K}]$ into $\GF^3[\KEQ{(K{+}1)}]$:
we relativise quantifiers with the coarsest distinguished symbol $E_{K+1}$ (this is the same technique as employed in Proposition~\ref{thm:GFequality}).
The symbol $E_{K+1}$ can now act as the universal relation,
and Lemma~\ref{l:GF3-reduction} produces a sentence which can be viewed as a $\GFU^3$-sentence.
The finite model property for the constant-free, equality-free $\GF^3[\EQ]$ follows then from the same property for the equality-free fragment of $\GFU^3$, which is established in~\cite{KR21}.

The described reduction yields decidability of satisfiability, but without tight complexity bounds for fragments $\GF^3[\KEQ{K}]$ parametrised by $K \ge 1$.
Indeed, Lemma~\ref{l:GF3-reduction} produces $\GFU^3$-sentences of $K$-exponential length, and the complexity of satisfiability for $\GFU^3$ is \NExpTime{}-complete~\cite{RS18}.
In consequence,
under standard complexity-theoretic assumptions,
we obtain only a non-optimal $(K{+}1)$-\NExpTime{} bound.
The optimal $(K{+}1)$-\ExpTime{} bound is established in the following subsection.

\subsection{Decision Procedure}\label{sec:GF3-decision-procedure}

\begin{proposition}\label{prop:GF3-decision-procedure}
  The satisfiability problem for the constant-free, equality-free $\GF^3[\KEQ{K}]$ is decidable in $(K{+}1)$-\ExpTime{}.
\end{proposition}

By Lemma~\ref{l:GF3-reduction},
to show the proposition, it suffices to design a $2$-\ExpTime{} procedure for $\GF^3[\KEQ{1}]$.
The idea employed here is inspired by Kiero\'nski and Malinowski's approach for $\GFU$~\cite{KM20},
as well as Pratt-Hartmanns's procedure for deciding the pure \GF{}~\cite{PH23}.

The procedure is based on a criterion that
expresses satisfiability of a sentence in terms of the existence of a certain set of \emph{labelled types}---let us first introduce this notion.

\subparagraph*{Labelled types.}
Let $\sigma$ be a constant-free signature, and let $\AAA$ denote the set of all $1$-types over $\sigma$.
Suppose that $\tau$ is a $k$-type over $\sigma$ for some $k \in \N$, and that $\ell \colon \operatorname{dom}(\tau) \to 2^{\AAA}$ assigns to each variable in $\operatorname{dom}(\tau) \eqdef \{x_1,\dots,x_k\}$, the \emph{domain} of $\tau$,
a set of $1$-types.
A pair $(\tau,\ell)$ is then called a \emph{labelled $k$-type} (over $\sigma$) if the following holds:
\begin{itemize}
  \item $\tau$ interprets $E_1$ as an equivalence relation. 
  \item For every $x_i \in \operatorname{dom}(\tau)$, we have $\type{\tau}{x_i} \in \ell(x_i)$.
  \item For every $x_i,x_j \in \operatorname{dom}(\tau)$, if $\tau \models x_i E_1 x_j$, then $\ell(x_i) = \ell(x_j)$.
\end{itemize}

For brevity, we sometimes write just $\tau$ for $(\tau,\ell)$ and refer to the associated labelling by~$\ell_\tau$.
A \emph{labelled type} is simply a labelled $k$-type for some $k$.

The intended role of the labelling $\ell_\tau$ is to capture information about $1$-types realised in the $E_1$-classes of elements of $\tau$: whenever $a_1,\dots,a_k$ are distinct elements of a structure $\str{A}$,
then the tuple $\langle a_1,\dots,a_k \rangle$ \emph{realises} the labelled $k$-type $(\tau,\ell_{\tau})$, where
\begin{itemize}
  \item $\tau = \type{\str{A}}{a_1,\dots,a_k}$;
  \item for every $i \in [k]$, $\ell_{\tau}(x_i) = \{ \type{\str{A}}{b} \mid b \in \absclass{\str{A}}{1}{a_i} \}$.
\end{itemize}

To illustrate, consider a structure $\str{A}$ over a domain $A = \{ 1,2,3,4 \}$.
Suppose that $E^{\str{A}}_1$ has three equivalence classes: $\{ 1,3 \}$, $\{ 2 \}$, and $\{ 4 \}$.
Then the labelled $2$-type of the pair $(1,2)$ is given by $(\type{\str{A}}{1,2}, \ell)$, where $\ell$ maps $x_1 \mapsto \{ \type{\str{A}}{1},\,\type{\str{A}}{3} \}$ and $x_2 \mapsto \{ \type{\str{A}}{2} \}$. 

Let $k,m \in \N$ be such that $1 \le k \le m$.
Suppose that $\tau$ is a labelled $k$-type and $\tau'$ is a labelled $m$-type.
We say that $\tau$ is a \emph{reduct} of $\tau'$ if
there exists an injective mapping $\rho \colon \operatorname{dom}(\tau) \hookrightarrow \operatorname{dom}(\tau')$
such that the following holds:
\begin{itemize}
  \item $\tau = \type{\tau'}{\rho(x_1),\dots,\rho(x_k)}$.
  \item $\ell_{\tau}(x_i) = \ell_{\tau'}(\rho(x_i))$ for every $x_i \in \operatorname{dom}(\tau)$.
\end{itemize}
Further, we say that $\tau'$ \emph{canonically extends} $\tau$ if the mapping $\rho$ is an inclusion, that is,
\begin{itemize}
  \item $\rho(x_i) = x_i$ for every $x_i \in \operatorname{dom}(\tau)$.
\end{itemize}

\subparagraph*{Satisfiability criterion.}
With necessary tools in place, we formulate the satisfiability criterion.

\begin{lemma}\label{l:GF3-criterion}
  Let $\phi$ be a constant-free, equality-free sentence of $\GF^3[\KEQ{1}]$,
  and let $\sigma$ be the signature of $\phi$.
  Assume that $\phi$ is in normal form as in Assumption~\ref{assumption:GF3-normal-form}, that is,
  \begin{flalign*}
    \phi
    \quad=\quad
    \Psi
    \quad\wedge\quad
    \bigwedge\nolimits_{t \in \cI}\forall x_1,x_2 \, \big(G_t(x_1,x_2) \rightarrow \exists y\, \psi_t(x_1,x_2,y)\big).
  \end{flalign*}

  Then $\phi$ is satisfiable if and only if there exists a non-empty set $\Omega$ of labelled types over $\sigma$ such that
\(
  P_{1,\Omega}(\tau), \dots, P_{5,\Omega}(\tau)
\)
hold for every $\tau \in \Omega$, where
\[
  P_{1,\Omega}(\tau)
  \iff
  \text{every reduct of $\tau$ belongs to $\Omega$.}
\]
\[
  P_{2,\Omega}(\tau)
  \iff
  \tau \models \Psi.
\]
\[
\begin{aligned}
  P_{3,\Omega}(\tau)
  \iff{}&
  \text{if $\tau$ is a labelled $1$-type and }
  \tau \models G_t(x_1,x_1)
  \text{ for some } t \in \cI,
  \\
  &\text{then there exists a labelled $2$-type }
  \tau' \in \Omega
  \text{ canonically extending $\tau$}
  \\
  &\text{such that }
  \tau' \models \psi_t(x_1,x_1,x_2).
\end{aligned}
\]
\[
\begin{aligned}
  P_{4,\Omega}(\tau)
  \iff{}&
  \text{if $\tau$ is a labelled $2$-type and }
  \tau \models G_t(x_1,x_2)
  \text{ for some } t \in \cI,
  \\
  &\text{then there exists a labelled $3$-type }
  \tau' \in \Omega
  \text{ canonically extending $\tau$}
  \\
  &\text{such that }
  \tau' \models \psi_t(x_1,x_2,x_3).
\end{aligned}
\]
\[
\begin{aligned}
  P_{5,\Omega}(\tau)
  \iff{}&
  \text{if $\tau$ is a labelled $1$-type, then for all }
  \alpha_1,\alpha_2 \in \ell_{\tau}(x_1),
  \\
  &\text{there exists a labelled $2$-type }
  \beta \in \Omega
  \text{ such that}
  \\
  &\type{\beta}{x_1} = \alpha_1,
  \qquad
  \type{\beta}{x_2} = \alpha_2,
  \\
  &\beta \models x_1 E_1 x_2,
  \qquad
  \ell_{\beta}(x_1)=\ell_{\beta}(x_2)=\ell_{\tau}(x_1).
\end{aligned}
\]
\end{lemma}

Intuitively, the conditions of Lemma~\ref{l:GF3-criterion} play the following roles.
The predicate $P_{1,\Omega}(\tau)$ expresses closure under reducts, a property satisfied by the set of types realised in any structure.
The predicate $P_{2,\Omega}(\tau)$ ensures satisfaction of the universal conjunct~$\Psi$ in any structure constructed from types in~$\Omega$.
The predicates $P_{3,\Omega}(\tau)$ and $P_{4,\Omega}(\tau)$ guarantee the existence of suitable witnessing extensions required by the existential conjuncts.
The predicate $P_{5,\Omega}(\tau)$ ensures that $E_1$-classes can be completed by providing the necessary labelled $2$-types.
Together, these conditions characterise the existence of a model of $\phi$.

\subparagraph*{Completeness.}
Let $\str{A} \models \phi$ be a model as in Lemma~\ref{l:propermodels}, and construct the set
\[ \Omega \;\eqdef\; \bigcup_{1 \le k \le 3}\big\{\, \text{$(\tau,\ell_{\tau})$ is a labelled $k$-type realised in $\str{A}$} \,\big\}. \]
One can verify that the conditions of Lemma~\ref{l:GF3-criterion} are satisfied for the constructed $\Omega$.

\subparagraph*{Soundness.}
Write $\AAA$ for the set of all $1$-types over the signature $\sigma$.

Suppose that $\Omega$ is a set of labelled types over $\sigma$ as in Lemma~\ref{l:GF3-criterion}.
We prove that $\phi$ is satisfiable by constructing a chain of finite structures
\[ \str{B}_1 \subseteq \str{B}_2 \subseteq \str{B}_3 \subseteq \dots, \quad\quad \text{where } \str{B}_{n+1} \,\restr\, B_n = \str{B}_n \text{ for all } n \ge 1. \]
The natural limit structure $\str{B}_{\omega} = \bigcup_{n \ge 1}{\str{B}_n}$ will be a model of $\phi$.

To guide the construction,
we define in parallel a sequence of functions $(\ell_n)_{n \in \N}$,
where each $\ell_n \colon B_n \rightarrow 2^{\AAA}$ assigns to every element of the domain of $\str{B}_n$ a set of $1$-types.
The label $\ell_n(a)$ serves as a \emph{declaration} of which $1$-types are admitted in the $E_1$-class of $a$.
The following conditions shall hold during the construction:
\begin{itemize}
  \item \( \{ \type{\str{B}_n}{b} \mid b \in \absclass{\str{B}_n}{1}{a} \} \subseteq \ell_n(a) \) for all $a \in B_n$.
  \item $\ell_n(a) = \ell_n(b)$ for all $(a,b) \in E^{\str{B}_n}_1$.
\end{itemize}
It is not required that all $1$-types from $\ell_n(a)$ are actually realised in the $E_1$-class of $a$.

We require each function $\ell_n$ to satisfy the following invariant:
\begin{enumerate}
  \item[($\ddagger$)]
  Let $(b_1,b_2) \in B_n^2$ be such that $\str{B}_n \models G_t(b_1,b_2)$ for some $t \in \cI$.
  \begin{itemize}
    \item If $b_1 = b_2$, then the labelled $1$-type $(\alpha,\ell_{\alpha})$ belongs to $\Omega$, where
    \[
      \alpha \eqdef \type{\str{B}_n}{b_1},
      \qquad
      \ell_{\alpha}(x_1) \eqdef \ell_n(b_1).
    \]
    \item If $b_1 \neq b_2$, then the labelled $2$-type $(\beta,\ell_{\beta})$ belongs to $\Omega$, where
    \[
      \beta \eqdef \type{\str{B}_n}{b_1,b_2},
      \qquad
      \ell_{\beta}(x_1) \eqdef \ell_n(b_1),
      \qquad
      \ell_{\beta}(x_2) \eqdef \ell_n(b_2).
    \]
  \end{itemize}
\end{enumerate}

\smallskip\noindent\emph{Base case.}
Since $\Omega$ is non-empty and $P_{1,\Omega}$ holds,
there exists a labelled $1$-type $(\alpha,\ell_{\alpha}) \in \Omega$.
Define $\str{B}_1$ to be the singleton structure with an element $b$ realising $\alpha$,
and set $\ell_1(b) \eqdef \ell_{\alpha}(x_1)$.

\smallskip\noindent\emph{Inductive step.}
Let $n \ge 1$.
Assume that $\str{B}_n$ and $\ell_n$ have already been defined.
Choose in a \emph{fair way} a pair $(b_1,b_2) \in B^2_n$
such that $\str{B}_n \models G_t(b_1,b_2)$, but $\str{B}_n \models \neg \exists y\; \psi_t(b_1,b_2,y)$ for some $t \in \cI$.
We consider only the case $b_1 \neq b_2$, as the case $b_1 = b_2$ is analogous.

Let $(\beta,\ell_{\beta})$ be the labelled $2$-type defined by
\[
  \beta \eqdef \type{\str{B}_n}{b_1,b_2},
  \qquad
  \ell_\beta(x_1) \eqdef \ell_n(b_1),
  \qquad
  \ell_\beta(x_2) \eqdef \ell_n(b_2).
\]
By the inductive invariant $(\ddagger)$, $(\beta,\ell_{\beta}) \in \Omega$.
Hence, by $P_{4,\Omega}$, there exists a labelled $3$-type $(\tau,\ell_{\tau}) \in \Omega$ that canonically extends $(\beta,\ell_{\beta})$ and satisfies $\tau \models \psi_t(x_1,x_2,x_3)$.

We now construct the structure $\str{B}_{n+1}$ by extending $\str{B}_n$ with a fresh element $b_3$.
\begin{itemize}
  \item Initialise \( \str{B}_{n+1} \eqdef \str{B}_n.\)
  \item The witnessing $3$-type: assign
  \( \type{\str{B}_{n+1}}{b_1,b_2,b_3} \eqdef \tau.\)
  (Since $\tau$ canonically extends $\beta$, this step does not modify the $2$-type $\type{\str{B}_n}{b_1,b_2}$.)
  \item Closure under~$E_1$:
  If either $(x_1,x_3)$ or $(x_2,x_3)$ belongs to $E_1^\tau$, then the corresponding transitive connections involving $b_3$ must be added.
  Assume w.l.o.g. that $(x_1,x_3) \in E_1^\tau$.

  For every
  \(b \in \absclass{\str{B}_n}{1}{b_1} \setminus \{b_1,b_2\},\)
  $P_{5,\Omega}$ yields a labelled $2$-type $(\gamma,\ell_\gamma) \in \Omega$ such that
  \[
  \begin{aligned}
    & \type{\gamma}{x_1} = \type{\str{B}_{n+1}}{b}, \qquad \type{\gamma}{x_2} = \type{\str{B}_{n+1}}{b_3}, \\
    & \gamma \models x_1 E_1 x_2, \quad\quad\qquad\;\;\;\; \ell_\gamma(x_1)=\ell_\gamma(x_2)=\ell_\tau(x_3).
  \end{aligned}
  \]
  We then set
  \(\type{\str{B}_{n+1}}{b,b_3} \eqdef \gamma.\)
\end{itemize}

We now have $\str{B}_{n+1} \models \psi_t(b_1,b_2,b_3)$.
To complete the inductive step, we declare
\[ \ell_{n+1} \;\eqdef\; \ell_n \cup \{ b_3 \mapsto \ell_{\tau}(x_3) \}.\]

\subparagraph*{Decision procedure.}
To conclude the proof, it remains to show that the existence of a set $\Omega$ satisfying the conditions of Lemma~\ref{l:GF3-criterion} can be decided in $2$-\ExpTime{}.
We call such a set a \emph{satisfiability witness}.
The decision procedure is given in Algorithm~\ref{algos}.

\begin{algorithm}[t]
  \caption{Decide satisfiability for constant-free, equality-free $\GF^3[\KEQ{1}]$}\label{algos}
  \KwIn{A normal-form sentence $\phi$ over a signature $\sigma$.}
  \KwOut{\textsc{Sat} if $\phi$ is satisfiable, and \textsc{UnSat} otherwise.}
  $\Omega \gets \{$all labelled $k$-types over $\sigma$ for $k \in \{1,2,3\}\}$\;
  \While{\rm there exists $\tau \in \Omega$ violating one of~$P_{1,\Omega}(\tau),\dots,P_{5,\Omega}(\tau)$}{
    choose such a $\tau$\;
    $\Omega \gets \Omega \setminus \{\tau\}$\;
  }
  \leIf{\rm $\Omega \neq \emptyset$}{\Return \textsc{Sat}}{\Return \textsc{UnSat}}
\end{algorithm}

The algorithm is clearly sound: it always terminates, and upon termination returns \textsc{Sat} if and only if the resulting set $\Omega$ is a satisfiability witness.

Completeness follows from the following observation.
Suppose that a satisfiability witness $\Omega^*$ exists.
Then no element of $\Omega^*$ is ever removed during the execution of the algorithm.
Consequently, upon termination, the algorithm produces a set $\Omega \supseteq \Omega^*$,
which is itself a satisfiability witness.
In particular, $\Omega$ is the unique maximal satisfiability witness.

To analyse the running time, observe that the procedure runs in polynomial time in the number of labelled types over at most three variables.
Each labelled type is a pair $(\tau,\ell_\tau)$, where $\tau$ is a $k$-type with $k \leq 3$ and
\(\ell_\tau \colon \operatorname{dom}(\tau) \to 2^{\AAA}.\)
The number of $k$-types with $k \leq 3$ is a priori doubly exponential in~$|\phi|$.
The number of possible labellings is bounded by
\((2^{|\AAA|})^3,\) which is likewise doubly exponential in~$|\phi|$;
indeed, since $\phi$ is constant-free, the number of $1$-types (= the size of $\AAA$) is only singly exponential in~$|\phi|$.
Hence, the total number of labelled $k$-types with $k \le 3$ is doubly exponential in~$|\phi|$, and the procedure runs in deterministic doubly-exponential time.
This concludes the proof of Proposition~\ref{prop:GF3-decision-procedure}.

\section{Related Work}\label{sec:related}

\subparagraph*{Two-Variable Fragment.}
With one equivalence relation, \FOt{} has the finite model property and \NExpTime-complete satisfiability~\cite{KO12}.
With two equivalence relations (without nesting condition), the finite model property is lost; both satisfiability and finite satisfiability are \TwoNExpTime-complete~\cite{KMPHT14}.
With three equivalence relations (again without nesting conditions), \FOt{} is undecidable~\cite{KO12}.

The work of~\cite{BB07} extends the results on \FOt{} with equivalences to the setting of words with \emph{nested data}, i.e., structures equipped with a linear order, nested equivalence relations, and additional free unary predicates.
It is shown there that satisfiability remains decidable when the linear order is accessible only through the successor relation, but becomes undecidable in the presence of both the order and the successor relation.

A limitation of~\cite{BB07} is the absence of free binary predicates.
This was recently addressed in~\cite{FKM25}, where it was shown that $\FOt{}[\EQ]$ enjoys the finite model property and has an \NExpTime-complete satisfiability problem.
Several further variants are investigated there:
\begin{itemize}
    \item nested equivalences combined with a linear order (without successor);
    \item nested total preorders (i.e., equivalence classes are linearly ordered);
    \item nested total preorders equipped with the induced successor relation.
\end{itemize}
These extensions retain elementary complexity, namely \NExpTime-completeness in the first two cases and \ExpSpace-completeness in the third.
Finally,~\cite{FKM25} establishes that \FOt{} with two independent families of nested equivalence relations is undecidable.

\subparagraph*{Logics for trees.}
An alternative perspective on nested equivalence relations is to interpret them as unranked trees of bounded depth.
When only $K$ distinguished predicates $E_1, \dots, E_K$ are considered, the corresponding structures can be viewed as describing the leaves of unranked trees of height $K{+}1$.
Domain elements correspond to leaves at level~$0$; the equivalence classes of $E_1$ form nodes at level~$1$; those of $E_2$ form nodes at level~$2$; and so on, up to level~$K$, with an implicit root located at level~$K{+}1$.
In this view, for every $1 \le k \le K$, the predicate $a E_k b$ holds precisely when $a$ and $b$ share a common ancestor at level~$k$.

This interpretation differs from standard logics over trees, where all tree nodes belong to the domain and the structure is navigated via predicates such as \emph{parent}, \emph{child}, or \emph{descendant}.
In this setting, \FOt{} was
investigated in numerous variants—over both ranked and unranked trees of unbounded depth, over data trees, and over trees with counting (see, e.g.,~\cite{BMS09,CW16b,BBC16}).

\subparagraph*{Equivalence guards.}
\GFt{} with two equivalence relations (without nesting condition) is $2$-\ExpTime-complete, but becomes undecidable with three equivalence relations~\cite{Kie05,KPT17}.
For an arbitrary number of equivalence relations, \GF{} remains decidable provided that equivalences are restricted to occur only in guards—so-called \emph{equivalence guards}~\cite{Kie05,KT18}.
However, in this setting the expressive power of the distinguished symbols is severely limited; in particular, one cannot express that equivalence relations are nested.

\subparagraph*{Other logics.}

Description logics with transitive roles ($\mathcal{S}$), role hierarchies ($\mathcal{H}$), and inverse roles ($\mathcal{I}$) are expressive enough to define nested equivalence relations.
Indeed, given transitive roles $T_1,T_2,\dots$, one can axiomatise them as nested equivalence relations using the inclusions
\[
  \top \sqsubseteq \exists T_k.\top,
  \qquad
  T_k \sqsubseteq T^{-1}_k,
  \qquad
  T_k \sqsubseteq T_{k+1}.
\]

Another formalism capable of defining families of nested equivalence relations is the Unary Negation Fragment extended with transitive roles, nominals, and role hierarchies.
The work of~\cite{DK19} showed that this logic is \TwoExpTime{}-complete.

\section{Concluding Remarks}\label{sec:remarks}

We conclude with brief observations on slight extensions of our results.

\subparagraph*{Cross products.}

A limitation of \GF{} is that it cannot express the following sentence:
\begin{equation*}
  \forall x,y\; \big(\big({\rm Elephant}(x) \wedge {\rm Mouse}(y)\big) \rightarrow {\rm biggerThan}(x,y)\big),
\end{equation*}
known in the context of description logics as a \emph{cross product}~\cite{RudolphKH08,BMP17,BORGIDA1996353}.
This sentence is, however, expressible in an extension of \GF{}, named the \emph{Triguarded Fragment} (\TGF{}), introduced in~\cite{RS18}.
\TGF{} extends the syntax of \GF{} (Definition~\ref{def:GFsyntax}) with the following rule:
\begin{itemize}
  \item If $\psi(x,y) \in \TGF$ has two free variables, then both $\forall x\,\psi(x,y)$ and $\exists x\,\psi(x,y)$ are in \TGF{}.
\end{itemize}

\TGF{} captures both \GF{} and \FOt{}, and in expressive power coincides with \GFU{} (defined in \autoref{sec:GF3-fmp}).
It turns out that our approach for $\GF[\EQ]$ extends to $\TGF[\EQ]$:
roughly speaking, the resulting bounds mirror those of Theorem~\ref{thm:GFmain},
but with deterministic complexity classes replaced by their nondeterministic counterparts (see \autoref{sec:TGF-complexity}).

\subparagraph*{Nominals.}

By Proposition~\ref{thm:GFequality},
for every $K \ge 1$, the satisfiability problem for \GF{}$[\KEQ{K}]$ with equality is undecidable.
Nevertheless, our proof techniques allow a restricted use of equality: the atoms $c = x$ and $c = c'$, where $x$ is a variable and $c$, $c'$ constants, pose no problem (see \autoref{sec:nominals}).
Even such a restricted use of equality has natural applications;
for example, in the standard translation of description logics with \emph{nominals} (singleton concepts).

\subparagraph*{Undecidable extensions.}

Our decidability results do not extend to the Loosely Guarded Fragment (\LGF{})~\cite{Gra99}.
In fact, an arbitrary \FO{} sentence
embeds into $\LGF[\KEQ{1}]$ by employing $\bigwedge_{x,x' \in \xs} E_1(x,x')$ as a loose guard for the quantifiers $\forall \xs'$ and $\exists \xs'$ with $\xs' \subseteq \xs$.

We cannot also support conjunctive queries: $\FO^2$ embeds into $\GF^2[\KEQ{1}]$; yet $\FO^2$ with (negations of) conjunctive queries is undecidable~\cite{Rosati07}.

\bibliography{master}

\appendix

\section{Missing Details from Section~\ref{sec:GF3-lower}: Succinct Formulas for 2-Numerals}\label{appendix:GF3-lower}

Fix $n \in \N$.
Recall that $2$-numerals are integers from the range $[0,2^{2^n}-1]$.
We show an encoding of $2$-numerals on the level of individual elements.
This encoding requires the presence of constants and an unbounded number of variables.
Importantly, the constructed formulas are of polynomial length in~$n$; more precisely of length $\cO(n^2)$.
The technique employed here is similar to one introduced in~\cite{RS18}.

Let $0$, $1$ be constant symbols which we identify with the usual numbers $0$ and $1$,
and let $BB$ be a predicate of arity $n+1$.
Given a structure $\str{A}$, we define an element $a \in A$ to \emph{represent} the $2$-numeral
\[ \num(a) \eqdef \sum_{k=0}^{2^n-1}{b_k \cdot 2^k}, \]
where $b_k = 1$ if
\( \str{A}\models BB(a,j_{n-1},\dots,j_0) \)
holds for the sequence $j_0,\dots,j_{n-1} \in \{ 0,1 \}$
encoding the binary representation of $k$, i.e., $k = \sum_{i=0}^{n-1}{j_i \cdot 2^i}$;
and otherwise $b_k = 0$.

Let $S$ and $Q$ be binary symbols and let $Z$ be unary.
We describe a theory such that its every model $\str{A}$ satisfies: if $(a,b) \in E^{\str{A}}_K$, then:
\begin{itemize}
    \item $\str{A} \models S(a,b)$ iff $\num(a) + 1 \equiv \num(b) \mod 2^{2^n}-1$.
    \item $\str{A} \models Z(a)$ iff $\num(a) = 0$.
    \item $\str{A} \models Q(a,b)$ iff $\num(a) = \num(b)$.
\end{itemize}

In what follows, we use variables $x$ and $y$,
as well as a ``generic'' tuple of variables denoted as $\bar{z} = z_1,\dots,z_k$;
the length $k \in \N$ is assumed to be such that the arities are respected.

We first axiomatise an auxiliary symbol $G$ of arity $n+1$:
$G(x,\bar{c})$ holds for every element $x$ and every sequence $\bar{c} \in \{ 0,1\}^n$.
We generate all possible sequences from the all-zero sequence:
\begin{flalign}
    &\forall x\, G(x,0^n) \; \wedge \;
    \bigwedge\nolimits_{i=1}^{n} \forall x,\bar{z}\, \big(G(x,\bar{z}) \rightarrow G(x,\bar{z}[z_i/1])\big)
\end{flalign}
Above, the notation $[z_i/1]$ means the substitution of the variable $z_i$ by the constant $1$.

Let $CB$ be a predicate of arity $n+1$. We axiomatise carry bits in an expected way:
\begin{flalign}
    &\forall x\, CB(x,0^n) \\
    &\bigwedge_{i=0}^{n-1} \, \forall x,\bar{z}\, \Big(
        G(x,\bar{z},0,1^i) \rightarrow
        \Big(CB(x,\bar{z},1,0^i) \leftrightarrow \big(CB(x,\bar{z},0,1^i) \wedge BB(x,\bar{z},0,1^i)\big)\Big)\Big)
\end{flalign}

We now axiomatise the successor predicate $S$.
Let $ChkS$ be a fresh predicate of arity $n+2$,
and let $\mu$ be a formula verifying that bit values agree with the carry:
\[ \mu(x,y,\bar{z}) \; \eqdef \; BB(y,\bar{z}) \leftrightarrow \big(BC(x,\bar{z}) \oplus BB(x,\bar{z})\big). \]

We express that if $x$ and $y$ are $E_K$-related and represent consecutive $2$-numerals,
then they are linked by the predicate $S$:
\begin{flalign}
  &\forall x,y\,\big(x E_K y \rightarrow ChkS(x,y,0^n)\big) \\
  &\bigwedge_{i=0}^{n-1}\; \forall x,y,\bar{z}\,\Big(\big(ChkS(x,y,\bar{z},0,1^i) \wedge \mu(x,y,\bar{z},0,1^i)\big) \rightarrow ChkS(x,y,\bar{z},1,0^i)\Big) \\
  &\forall x,y\,\big(\big(ChkS(x,y,1^n) \wedge \mu(x,y,1^n)\big)\rightarrow S(x,y)\big)
\end{flalign}

For converse, we verify that any $S$-successive elements represent consecutive $2$-numerals:
\begin{flalign}
  &\forall x,y\, \big(S(x,y) \rightarrow \big(\mu(x,y,1^n) \wedge ChkS(x,y,1^n)\big)\big) \\
  &\bigwedge_{i=0}^{n-1}\; \forall x,y,\bar{z}\,\Big(ChkS(x,y,\bar{z},1,0^i) \rightarrow \big(\mu(x,y,\bar{z},0,1^i) \wedge ChkS(x,y,\bar{z},0,1^i)\big)\Big)
\end{flalign}

The axiomatisation of successor is now finished; equality and zero tests are treated in a similar way.

\section{Missing Details from Section~\ref{sec:GF-upper}: Proofs for the General Case}\label{appendix:GF-upper}

In \autoref{sec:GF-upper}, for simplicity, we worked with the three-variable fragment without constants.
In this appendix, we remove this assumption and establish the results of \autoref{sec:GF-upper} in full generality.
As a consequence, we obtain Theorems~\ref{thm:GF-fmp} and~\ref{thm:GFmain}.

We assume familiarity with Section~\ref{sec:GF-upper}.
The overall ideas remain unchanged, and most definitions and claims generalise in a straightforward manner.
Accordingly, we omit detailed explanations and focus only on the necessary modifications.

\subsection{Preparations}

\subparagraph*{Normal form.}
A $\GF[\KEQ{K}]$-sentence $\phi$ is in \emph{normal form} if it has the shape:
\begin{flalign}\label{GF:normal-form-shape}
  \phi \quad=\quad \Psi \quad\wedge\quad \bigwedge\nolimits_{t \in \cI} \forall \xs\, \big(G_t(\xs) \rightarrow \exists y\, \psi_t(\xs,y)\big),
\end{flalign}
where $\Psi$ is a conjunction of purely universal prenex sentences, each $G_t(\xs)$ is a guard (possibly being a distinguished predicate), and each $\psi_t(\xs,y)$ is quantifier-free.
Here, $\xs$ is a tuple of variables, whereas $y$ is a single variable.
To simplify notation, all conjuncts have a common quantifier prefix ``$\forall \xs$''.
We do not insist that all variables of $\xs$ are actually utilised by $G_t(\xs)$.
Nevertheless, we require that $\FreeVars(\psi_t) \setminus \{ y \} \subseteq \FreeVars(G_t) \subseteq \xs$.

The standard normal-form reduction (see, e.g.,~\cite{Gra99,PH23}) remains sound for the fragments $\GF[\KEQ{K}]$, as well as for their constant-free and finite-variable subfragments.

\begin{lemma}\label{l:normal-form}
  For every equality-free sentence $\phi\in\GF[\KEQ{K}]$ over a signature $\sigma$,
  there exists a normal-form equality-free sentence $\phi_{\rm nf}\in\GF[\KEQ{K}]$ over an extended signature $\sigma_{\rm nf} \supseteq \sigma$ such that:
  \begin{itemize}
    \item $\phi$ and $\phi_{\rm nf}$ are equisatisfiable.
    \item The $\sigma$-reduct of any model $\str{A}\models\phi_{\rm nf}$ satisfies $\phi$.
    \item $\phi_{\rm nf}$ is computable in polynomial time from $\phi$.
  \end{itemize}
\end{lemma}

\subparagraph*{Standard name assumption.}
We work under the \emph{standard name assumption}, that is, we assume that constant symbols are interpreted in structures by themselves.
In particular, two different constants are interpreted by different elements.
W.r.t.~satisfiability, this assumption is~w.l.o.g.: whenever $\phi$ has a model $\fA$ such that $c_1^{\fA} = c_2^{\fA}$, then we replace occurrences of $c_2$ in $\phi$ with $c_1$ without affecting satisfiability.
Given a structure $\fA$, we partition the domain $A = A_0 \uplus \Cons{}$, where $A_0$ and $\Cons$ are disjoint sets of \emph{unnamed} elements and constants, respectively.

\subparagraph*{Types.}
Let $\sigma$ be a signature.
For $k \in \N$, let $\operatorname{Lit}_k(\sigma)$ denote the set of all literals over relation and constant symbols from $\sigma$ and variables $x_1,\dots,x_k$; literals with ``$=$'' are not included.
A \emph{$k$-type} over $\sigma$ is a maximal consistent subset of $\operatorname{Lit}_k(\sigma)$.
A \emph{type} is a $k$-type for some $k \ge 1$.

A $k$-type $\tau$ is \emph{guarded} if there exists a positive literal $\gamma \in \tau$ such that $\FreeVars(\gamma) = \{ x_1,\dots,x_k \}$, that is, if some positive literal uses all variables of $\tau$.

Given a $k$-tuple of distinct \emph{unnamed} elements $\langle a_1,\dots,a_k \rangle$ of a $\sigma$-structure $\fA$, we write 
\( \type{\fA}{a_1,\dots,a_k} \) 
for the $k$-type \emph{realised} by $\langle a_1,\dots,a_k \rangle$ in $\fA$.  
Formally, $\type{\fA}{a_1,\dots,a_k}$ is the set of literals $\gamma \in \operatorname{Lit}_k(\sigma)$ such that $\fA,f \models \gamma$, where $f$ sends $x_i \mapsto a_i$ for all $i \in [k]$.

We view $k$-types as $\sigma$-structures over the canonical domain $\{x_1,\dots,x_k\} \cup \Cons$, where $\Cons \subseteq \sigma$ is the set of constants. We thus allow structure-like notation also for types.

\subsection{Finite-Index Nesting Property}\label{sec:proof-of-GF-fin-property}

\begin{lemma}\label{l:GF-fin-property}
  Let $K \ge 2$, and let $\phi$ be an equality-free sentence of $\GF[\KEQ{K}]$ in normal form.
  Let $m \in \N$ denote the number of variables in~$\phi$,
  and let $\AAA$ denote the set of all $1$-types over the signature of~$\phi$.
  If $\phi$ is satisfiable, then it has a model $\fA \models \phi$ such that
  \[
    |\cC / E^{\fA}_1| \;\le\; m \cdot 2^{|\AAA|}
    \qquad\text{for every}\qquad
    \cC \in A / E^{\fA}_2.
  \]
\end{lemma}

\subparagraph*{Proof of Lemma~\ref{l:GF-fin-property}.}

Fix $K \ge 2$, and let $\phi$ be a sentence of $\GF[\KEQ{K}]$ in the shape given by \eqref{GF:normal-form-shape}.
Let $\sigma$ denote the signature of $\phi$, and let $\Cons \subseteq \sigma$ denote the set of constants.
Write $m$ for the number of variables occurring in $\phi$, and $\AAA$ for the set of $1$-types over $\sigma$.

\begin{lemma}\label{l:propermodels-full}
  If $\phi$ is satisfiable,
  then it has a model~$\str{A} \models \phi$ with the following properties.
  \begin{enumerate}
    \item\label{l:propermodelsitem1-full} For every assignment $f\colon \xs \rightarrow A$ and every index $t \in \cI$, whenever $\str{A},f \models G_t$,
    then there exists $a \in A \setminus (f(\xs) \cup \Cons)$ such that
    \[ \str{A}, f \cup \{ y \mapsto a \} \models \psi_t. \]
    \item\label{l:propermodelsitem2-full} If $\cC \in A/E^{\str{A}}_1$ and $\alpha \in \AAA$, then the set
    \[ \{ a \in \cC \setminus \Cons \mid \type{\str{A}}{a} = \alpha \} \]
    is either empty or countably infinite.
    \item\label{l:propermodelsitem3-full}
    For every $\cC \in A/E^{\str{A}}_1$, we have $\cC \setminus \Cons \neq \emptyset$.
  \end{enumerate}
\end{lemma}
\begin{proof} The lemma follows by the technique of Lemma~\ref{l:propermodels}. \end{proof}

Assuming that $\phi$ is satisfiable, we fix a model $\str{A} \models \phi$ as in Lemma~\ref{l:propermodels-full},
and subsequently, construct a new model $\str{B} \models \phi$ that will satisfy Lemma~\ref{l:GF-fin-property}.

\subparagraph*{Sorts and spectra.}
We adapt the definitions from \autoref{sec:GF3-fin-property}.
\begin{itemize}
  \item A \emph{class-sort} of an element $a \in A$ is the following set of $1$-types:
  \[ \classsort{\str{A}}{a} \eqdef \big\{ \type{\str{A}}{b} \bigmid b \in \absclass{\str{A}}{1}{a} \setminus \Cons \big\} \subseteq \AAA. \]
  By item~\autoref{l:propermodelsitem3-full} of Lemma~\ref{l:propermodels-full}, the above set is non-empty.

  We also define $\classsort{\str{A}}{\cC} \eqdef \classsort{\str{A}}{a}$, where
  $\cC \in A/E^{\str{A}}_1$ and $a \in \cC$ is any representative.
  For $\Gamma \subseteq \AAA$, the $E_1$-class $\cC$ is said to \emph{realise} $\Gamma$ if $\classsort{\str{A}}{\cC} = \Gamma$.

  \item Let $S \subseteq A$ be such that $S/E^{\str{A}}_1 \subseteq A/E^{\str{A}}_1$.
  A \emph{class-spectrum} of $S$ is the function $\classspectrum{\str{A}}{S} \colon 2^{\AAA} \rightarrow \N$ that sends each $\Gamma \subseteq \AAA$ to
\[ \classspectrum{\str{A}}{S}(\Gamma) \eqdef \big|\big\{ \cC \in S/E^{\str{A}}_1 \bigmid \classsort{\str{A}}{\cC} = \Gamma \big\}\big|. \]
\end{itemize}

\subparagraph*{Truncation.}
Let $\cC \in A/E^{\str{A}}_{2}$.
A subset $T \subseteq \cC$ is an \emph{$m$-truncation} of $\cC$ if
\begin{romanenumerate}
  \item $T/E^{\str{A}}_1 \subseteq \cC/E^{\str{A}}_1$, and 
  \item for every $\Gamma \subseteq \AAA$, we have
  \[ \classspectrum{\str{A}}{T}(\Gamma) \;=\; \inf\big\{\,\classspectrum{\str{A}}{\cC}(\Gamma),\, m \,\big\}. \]
\end{romanenumerate}

For every $\cC \in A/E^{\str{A}}_2$, fix an $m$-truncation $T_m(\cC)$. Define
\[B \eqdef \bigcup_{\cC \in A/E^{\str{A}}_2}T_m(\cC).\]

\subparagraph*{Similar tuples.}
Let $k \in [0,m]$, and let $\bar{a} = (a_1,\dots,a_k)$ and $\bar{b} = (b_1,\dots,b_k)$ be tuples of distinct elements from~$A \setminus \Cons$.
We write $\bar{a} \sim \bar{b}$ if the following conditions hold: 
\begin{romanenumerate}
  \item \label{patternitem1-full} For $i \in [k]$, $\type{\str{A}}{a_i} = \type{\str{A}}{b_i}$.
  \item \label{patternitem2-full} For $i \in [k]$, $\classsort{\str{A}}{a_i} = \classsort{\str{A}}{b_i}$.
  \item For $i,j \in [k]$, $(a_i,a_j) \in E^{\str{A}}_1$ if and only if $(b_i,b_j) \in E^{\str{A}}_1$.
  \item \label{patternitem3-full} For $i \in [k]$ and $\ell \in [2,K]$, $(a_i,b_i) \in E^{\str{A}}_\ell$.
\end{romanenumerate}

\begin{lemma}\label{l:patterns-full}
  Let $k \in [0,m{-}1]$, and let $\bar{a} = (a_1,\dots,a_k)$ and $\bar{b} = (b_1,\dots,b_k)$ be tuples of distinct elements from $A \setminus \Cons$ and from $B \setminus \Cons$, respectively.
  Suppose that
  \[ \text{$\bar{a} \sim \bar{b}$ \quad and \quad $F \subseteq A$ is a finite set with $\bar{b} \subseteq F$.}\]
  Then, for any $a_{k+1} \in A \setminus (\bar{a} \cup \Cons)$,
  one can find $b_{k+1} \in B \setminus (F \cup \Cons)$ such that
  \[ (a_1,\dots,a_k,a_{k+1}) \sim (b_1,\dots,b_k,b_{k+1}). \]
\end{lemma}
\begin{proof}
  The lemma follows by generalising the strategy from the proof of Lemma~\ref{l:patterns}.
\end{proof}

\subparagraph*{Model construction.}
We now construct a chain of finite structures
\[ \str{B}_1 \subseteq \str{B}_2 \subseteq \str{B}_3 \subseteq \dots, \quad\quad \text{where } \str{B}_{n+1} \,\restr\, B_n = \str{B}_n \text{ for all } n \ge 1. \]
The natural limit structure $\str{B}_{\omega} = \bigcup_{n \ge 1}{\str{B}_n}$ will be a model for $\phi$ satisfying Lemma~\ref{l:GF-fin-property}.

For every $n \ge 1$, the structure $\str{B}_n$ will be defined over a subset $B_n \subseteq B$, and
the inclusion mapping $\str{B}_n \hookrightarrow \str{A}$ will preserve the $1$-types and the nested equivalence relations; though not necessarily other atoms. Formally:
\begin{itemize}
  \item For every $b \in B_n \setminus \Cons$, $\type{\str{B}_n}{b} = \type{\str{A}}{b}$.
  \item For every $k \in [K]$, $E^{\str{B}_n}_k = E^{\str{A}}_k \cap B^2_n$.
\end{itemize}

In parallel, we define a sequence of mappings $(\rho_n)_{n \ge 1}$:
\begin{itemize}
  \item Let $n \ge 1$, and let $\bar{b}$ be a tuple of distinct elements from $B_n \setminus \Cons$ such that 
  the type $\type{\str{B}_n}{\bar{b}}$ is guarded.
  Then the mapping $\rho_n$ sends $\bar{b}$ to a tuple $\bar{a}$ of distinct elements from $A \setminus \Cons$ such that
  \[ |\bar{a}| = |\bar{b}|, \quad \bar{a} \sim \bar{b}, \quad\text{and}\quad \type{\fA}{\bar{a}} = \type{\str{B}_n}{\bar{b}}. \]
\end{itemize}

Let $f\colon \xs \rightarrow B_n$ be an assignment, and let $t \in \cI$ be the index of a conjunct of $\phi$.
We say that the pair $(f,t)$ is \emph{satisfied} in~$\str{B}_n$ if $\str{B}_n,f \models G_t \rightarrow \exists y\;\psi_t$, and otherwise it is \emph{unsatisfied}.

\begin{claim}
We have $\fB_{\omega} \models \phi$ whenever the following holds:
\begin{enumerate}
  \item For every $n \ge 1$, the structure $\str{B}_n$ interprets $E_1,\dots,E_K$ as nested equivalence relations.
  \item For every $n \ge 1$, we have $\str{B}_n \models \Psi$.
  \item Every pair $(f,t)$, where $f \colon \xs \rightarrow B_n$ for some $n \ge 1$ and $t \in \cI$,
  eventually becomes satisfied in some structure $\str{B}_{\ell}$ with $\ell \ge n$.
\end{enumerate}
\end{claim}

\smallskip\noindent\emph{Base case.}
Let $a \in B \setminus \Cons$ be any unnamed element.
Define $\str{B}_1 \eqdef \str{A}~\restr~(\{a\} \cup \Cons)$.
The corresponding mapping is $\rho_1 \eqdef \{ () \mapsto (), (a) \mapsto (a) \}$.

\smallskip\noindent\emph{Inductive step.}
Let $n \ge 1$.
Assume that $\str{B}_n$ and $\rho_n$ have already been defined.
If there are no unsatisfied pairs in $\str{B}_n$, we declare $\str{B}_{\omega} \eqdef \str{B}_n$ as the final model.
Yet, suppose that an unsatisfied pair exists:
choose such a pair $(f,t)$, where $f\colon \xs \rightarrow B_n$ and $t \in \cI$.
Naturally, it must be chosen in a \emph{fair way}: no pair is omitted infinitely often.

We construct now the next structure $\str{B}_{n+1}$ by extending $\str{B}_n$ with one fresh element and making the pair $(f,t)$ satisfied.

Let $\bar{b} = (b_1,\dots,b_k)$ enumerate $f(\xs) \setminus \Cons$.
W.l.o.g.~$\{ b_1,\dots,b_k \} = f(\FreeVars(G_t)) \setminus \Cons$;
that is, we discard elements that are not in the scope of the guard $G_t$ under $f$.

Write $\bar{a} = (a_1,\dots,a_k)$ for $\rho_n(\bar{b})$.
By the inductive invariant of $\rho_n$, we have
\[ |\bar{a}| = |\bar{b}|, \quad \bar{a} \sim \bar{b}, \quad\text{and}\quad \type{\fA}{\bar{a}} = \type{\str{B}_n}{\bar{b}}. \]

Define an assignment $g\colon \xs \rightarrow A$ in parallel to $f$,
that is,
\[
  g(x) = \begin{cases}
    \, a_i \quad & \text{if} \quad f(x) = b_i, \\
    \, f(x) \quad & \text{if} \quad f(x) \in \Cons.
  \end{cases}
\]

Since $\str{A} \models \phi$, and in particular $\str{A},g \models G_t$,
there exists a witness $a_{k+1} \in A$ for the pair $(g,t)$, that is, an element satisfying
\[ \str{A},g \cup \{ y \mapsto a_{k+1} \} \models \psi_t. \]
By item~\ref{l:propermodelsitem1-full} of Lemma~\ref{l:propermodels-full}, we take $a_{k+1} \in A \setminus (\bar{a} \cup \Cons)$.
Applying Lemma~\ref{l:patterns-full} with $F \eqdef B_n$, we obtain $b_{k+1} \in B \setminus B_n$ such that
\[ (a_1,\dots,a_{k+1}) \sim (b_1,\dots,b_{k+1}).\]

We now construct the structure $\str{B}_{n+1}$ by extending $\str{B}_{n}$ with the element $b_{k+1}$.
\begin{itemize}
  \item Initialise $\str{B}_{n+1} \eqdef \str{B}_{n}$,
  \item Assign the $(k{+}1)$-type of $(b_1,\dots,b_{k+1})$: \[\type{\str{B}_{n+1}}{b_1,\dots,b_{k+1}} \eqdef \type{\str{A}}{a_1,\dots,a_{k+1}}.\]
  Since $\type{\str{B}}{\bar{b}} = \type{\str{A}}{\bar{a}}$, this step ensures that $\str{B}_{n+1} \restr B_n = \str{B}_n$.
  \item Define the remaining $2$-types:
  for every $b \in B_n \setminus (\bar{b} \cup \Cons)$ with $(b,b_{k+1}) \in E^{\str{A}}_K$, set
  \[ \type{\str{B}_{n+1}}{b,b_{k+1}} \eqdef \type{\str{A}}{b,b_{k+1}}.\]
  Again, since $\type{\str{B}}{b} = \type{\str{A}}{b}$, this does not modify $\str{B}_{n+1} \restr B_n$.
  \item The facts whose truth values remain unspecified are set to false.
\end{itemize}

The construction of $\str{B}_{n+1}$ is now finished. We next extend the mapping $\rho_n$ to $\rho_{n+1}$.
\begin{itemize}
  \item Initialise $\rho_{n+1} \eqdef \rho_n$.
  \item For every $\ell \in [0,k]$ and every choice $j_1,\dots,j_{\ell} \in [k]$ of pairwise distinct indices,
  whenever $\type{\str{B}_{n+1}}{b_{j_1},\dots,b_{j_{\ell}},b_{k+1}}$ is guarded, set
  \[ \rho_{n+1}(b_{j_1},\dots,b_{j_t},b_{k+1}) \eqdef (a_{j_1},\dots,a_{j_t},a_{k+1}).\]
  \item For every $b \in B_n \setminus (\bar{b} \cup \Cons)$
  such that $\type{\str{B}_{n+1}}{b,b_{k+1}}$ is guarded, set
  \[ \rho_{n+1}(b,b_{k+1}) \eqdef (b,b_{k+1}) \quad\text{and}\quad \rho_{n+1}(b_{k+1},b) \eqdef (b_{k+1},b).\]
\end{itemize}

This concludes the inductive step. The following claim serves as a brief summary of the most important inductive invariants.

\begin{claim}\label{constructioninvariantclaim-full}
  \,
  \begin{enumerate}
    \item\label{invariantclaim1-full} $\str{B}_{n+1}~\restr~B_n = \str{B}_n$ and $B_{n+1} \setminus B_n = \{ b_{k+1} \}$.
    \item\label{invariantclaim2-full} The inclusion $\str{B}_{n+1} \hookrightarrow \str{A}$ preserves the $1$-types and the distinguished predicates $E_1,\dots,E_K$; that is, $\type{\fB_{n+1}}{b} = \type{\fA}{b}$ for every $b \in B_{n+1}$ and $E_k^{\str{B}_{n+1}} = E^{\str{A}}_k \cap B^2_{n+1}$ for every $k \in [K]$.
    \item\label{invariantclaim3-full} For every tuple $\bar{b}$ of distinct elements from $B_{n+1} \setminus \Cons$ such that $\type{\str{B}_{n+1}}{\bar{b}}$ is guarded,
    we have
    $\rho_{n+1}(\bar{b}) \sim \bar{b}$ and $\type{\str{A}}{\rho_{n+1}(\bar{b})} = \type{\str{B}_{n+1}}{\bar{b}}$.
    \item $\fB_{n+1} \models \Psi$.
    \item $\fB_{n+1},f \cup \{ y \mapsto b_{k+1} \} \models \psi_t$.
  \end{enumerate}
\end{claim}

\subparagraph*{Bounding index.}

Since at most $m$ realisations of every class-sort $\Gamma \subseteq \AAA$ are included in $T_m(\cC)$,
we obtain \[ |T_m(\cC)/E^{\str{A}}_1| \le m \cdot  2^{|\AAA|}.\]
The desired bound on index follows, as the inclusion mapping $\str{B}_{\omega} \hookrightarrow \str{A} \restr B_{\omega}$ preserves both $E_1$ and $E_2$.
This concludes the proof of Lemma~\ref{l:GF-fin-property}.

\subsection{Elimination of Distinguished Symbols}\label{sec:succinct-reduction}

\begin{lemma}\label{l:succinct-reduction}
  Let $K \ge 2$,
  and let $\phi$ be an equality-free sentence of $\GF[\KEQ{K}]$.
  Then there exists an equality-free sentence $\phi'$ of $\GF[\KEQ{1}]$ such that the following holds.
  \begin{itemize}
    \item The sentences $\phi$ and $\phi'$ are equisatisfiable.
    \item If $\phi'$ has a model over a domain $D$, then $\phi$ has also.
    \item $|\phi'|$ is bounded by a $(K-1)$-exponential function of $|\phi|$.
    \item $\phi'$ is computable from $\phi$ in polynomial time in $|\phi'|$.
  \end{itemize}
\end{lemma}
\begin{proof}
Let $K \ge 2$,
and let $\phi \in \GF{}[\KEQ{K}]$ be an equality-free sentence.
By Lemma~\ref{l:normal-form}, we assume that $\phi$ is in normal form.
Write $\AAA$ for the set of $1$-types over the signature of $\phi$.

Set $n = \lceil \log_2 N \rceil$, where $N = m \cdot 2^{|\AAA|}$ is the bound of Lemma~\ref{l:GF-fin-property}.

In the proof of Lemma~\ref{l:GF3-reduction},
we gave a strategy to replace the predicate $E_1$
with a family of unary predicates $\{ U_i \}^{n}_{i=1}$.
This strategy remains sound also for the current setting,
and thus Lemma~\ref{l:GF3-reduction} establishes a reduction from $\GF[\KEQ{K}]$ to $\GF[\KEQ{(K{-}1)}]$,
and by iterative applications, to $\GF[\KEQ{1}]$.
However, this reduction is inefficient for $\GF[\KEQ{K}]$ with constants,
as it yields sentences of $(2K-2)$-exponential length.

Indeed: elimination of one distinguished symbol causes the length of $\phi$ to increase by $\cO(|\phi|+|\AAA|)$.
With unbounded-arity symbols and constants, the size of $\AAA$ is doubly exponential in~$|\phi|$.
Hence, every elimination step increases length doubly exponentially.

To meet the requirement on the length, we improve the reduction to increase length only at exponential rate.
The technique is to employ the succinct axiomatisation of $2$-numerals, i.e., integers from $[0,2^{2^n}-1]$, instead of unary-coding via predicates $\{ U_i \}^{n}_{i=1}$.
We change axiom~\eqref{naive-reduction-axiom2} from the proof of Lemma~\ref{l:GF3-reduction} to
\begin{flalign}
  \mu_n \; \wedge \; \forall x_1,x_2\;\Big(&x_1 E_2 x_2 \rightarrow \Big(x_1 E_1 x_2 \leftrightarrow Q(x_1,x_2)\Big)\Big),
\end{flalign}
where the sentence $\mu_n$ axiomatises the theory of $2$-numerals (see Appendix~\ref{appendix:GF3-lower} for construction), and $Q$ is the equality test: $Q(x_1,x_2)$ holds iff the elements $x_1,x_2$ with $(x_1,x_2) \in E_2$ represent the same $2$-numeral.
The rest of the proof stays as in Lemma~\ref{l:GF3-reduction}.
\end{proof}

\subsection{Finite Model Property}\label{sec:FMP-full}

We establish now the finite model property for the equality-free $\GF[\EQ]$, together with tight bounds on minimal model size.

Let $\GFU$ denote the fragment $\GF$ extended with a single distinguished binary symbol~$U$.
The interpretation of $U$ is constrained to be the universal relation. 

The finite model property for $\GF[\EQ]$ follows from the same property for \GFU{}:
we first embed $\GF[\KEQ{K}]$ into $\GF[\KEQ{(K{+}1)}]$ by relativising quantifiers with the additional distinguished symbol $E_{K+1}$.
The reduction of Lemma~\ref{l:succinct-reduction} yields then a $\GF[\KEQ{1}]$-sentence of $K$-exponential length, which can be viewed as a $\GFU$-sentence.
It is shown in~\cite{KR21} that every satisfiable equality-free sentence of \GFU{} admits a doubly exponential model.
Thus, we obtain a $(K{+}2)$-exponential upper bound on the size of minimal models.

The fixed-variable fragment $\GF^m[\KEQ{K}]$ follows the same scheme: for every $m \in \N$, satisfiable sentences of $\GFU^m$ admit smaller models of singly exponential size.
This yields a $(K{+}1)$-exponential upper bound on minimal model size for equality-free $\GF^m[\KEQ{K}]$.

The constant-free $\GF[\KEQ{K}]$ requires, however, a different treatment, as $\GFU$ without constants can still enforce models of doubly exponential size.
We achieve the tight bound on model size as follows:
given a sentence of the constant-free $\GF[\KEQ{K}]$,
we first reduce it to a sentence in $\GF[\KEQ{(K{-}1)}]$ with constant; the reduction is as described in \autoref{sec:succinct-reduction}.
This reduction incurs only a polynomial increase in the sentence length.
Indeed, without constants the number of $1$-types is only singly exponential,
and hence the bound on nesting-index (Lemma~\ref{l:GF-fin-property}) is doubly exponential;
so succinct formulas of polynomial length from Appendix~\ref{appendix:GF3-lower} suffice for the reduction.
The result follows then from the $(K{+}1)$-exponential upper bound already shown for $\GF[\KEQ{(K{-}1)}]$ with constants.

\subsection{Decision Procedure}\label{sec:decision-procedure}

\begin{proposition}
  The satisfiability problem for the equality-free $\GF[\KEQ{1}]$ is decidable in $3$-\ExpTime{}.
\end{proposition}

The procedure is based on a satisfiability criterion (Lemma~\ref{l:decision-procedure}),
expressing satisfiability of a sentence in terms of the existence of a certain set of \emph{labelled types}.

\subparagraph*{Labelled types.}
Let $\sigma$ be a signature with $\Cons \subseteq \sigma$ being the set of constants.
Write $\AAA$ for the set of $1$-types over $\sigma$.
Suppose that $\tau$ is a $k$-type over $\sigma$ for some $k \in \N$, and that $\ell \colon \operatorname{dom}(\tau) \to 2^{\AAA}$ is a function, where $\operatorname{dom}(\tau) \eqdef \{ x_1,\dots,x_k \} \cup \Cons$ is the \emph{domain} of $\tau$.
A pair $(\tau,\ell)$ is called a \emph{labelled $k$-type} (over $\sigma$) if
there exists a $\sigma$-structure $\str{A}$ over a domain $A$ with $\operatorname{dom}(\tau) \subseteq A$ such that the following conditions hold:
\begin{itemize}
  \item $\type{\str{A}}{x_1,\dots,x_k} = \tau$.
  \item For every $x \in \operatorname{dom}(\tau)$,
  \( \ell(x) = \{ \type{\str{A}}{a} \bigmid a \in \absclass{\str{A}}{1}{x} \setminus \Cons \}. \)
\end{itemize} 

We sometimes write just $\tau$ for $(\tau,\ell)$ and refer to the associated labelling by~$\ell_\tau$.
A \emph{labelled type} is simply a labelled $k$-type for some $k$.

Suppose that $\tau$ is a labelled $k$-type and $\tau'$ is a labelled $m$-type, where $0 \le k \le m$.
We say that $\tau$ is a \emph{reduct} of $\tau'$ if
there exists an injective mapping $\rho \colon \operatorname{dom}(\tau) \hookrightarrow \operatorname{dom}(\tau')$
such that:
\begin{itemize}
  \item $\rho$ is the identity on $\Cons$ and injectively sends variables of $\tau$ to variables of $\tau'$, that is,
  \begin{itemize}
    \item for every $c \in \Cons$, $\rho(c) = c$;
    \item for every $i \in [k]$, $\rho(x_i) = x_j$ for some $j \in [m]$;
    \item for all $1 \le i \le j \le k$, $\rho(x_i) \neq \rho(x_j)$.
  \end{itemize}
  \item $\tau = \type{\tau'}{\rho(x_1),\dots,\rho(x_k)}$.
  \item $\ell_{\tau}(x) = \ell_{\tau'}(\rho(x))$ for every $x \in \operatorname{dom}(\tau)$.
\end{itemize}

Further, we say that $\tau'$ \emph{canonically extends} $\tau$, and write $\tau' \models \tau$,
if $\rho$ is the inclusion mapping, that is, $\rho(x) = x$ for all $x \in \operatorname{dom}(\tau)$.

\subparagraph*{Satisfiability criterion.}
Lemma~\ref{l:decision-procedure} establishes the desired criterion for satisfiability.

\begin{lemma}\label{l:decision-procedure}
  Let $\phi$ be an equality-free sentence of $\GF[\KEQ{1}]$ in normal form, that is,
  \[ \phi \quad=\quad \Psi \quad\wedge\quad \bigwedge\nolimits_{t \in \cI} \forall \xs\, \big(G_t(\xs) \rightarrow \exists y\, \psi_t(\xs,y)\big). \]
  Write $\sigma$ for the signature of $\phi$, and let $\Cons \subseteq \sigma$ be the set of constants.

  Then $\phi$ is satisfiable if and only if there exists a set $\Omega$ of labelled types over $\sigma$ such that
  $P_{0,\Omega}$ holds, and
\(
  P_{1,\Omega}(\tau), \dots, P_{4,\Omega}(\tau)
\)
hold for every $\tau \in \Omega$, where
\[
  P_{0,\Omega}{\color{white}(\tau)}
  \iff
  \text{$|\Omega| \ge 2$ and $\Omega$ induces a unique $0$-type.}
\]
\[
  P_{1,\Omega}(\tau)
  \iff
  \text{every reduct of $\tau$ belongs to $\Omega$.}
\]
\[
  P_{2,\Omega}(\tau)
  \iff
  \tau \models \Psi.
\]
\[
\begin{aligned}
  P_{3,\Omega}(\tau)
  \iff{}&
  \text{if $\tau$ is a labelled $k$-type and $f \colon \xs \to \operatorname{dom}(\tau)$ is such that }
  \\
  &\operatorname{dom}(\tau) \subseteq f(\xs) \cup \Cons \;\text{ and }\;
  \tau,f \models G_t(\xs) \text{ for some } t \in \cI,
  \\
  &\text{then there exists a labelled $(k{+}1)$-type }
  \tau' \in \Omega
  \text{ canonically extending $\tau$}
  \\
  &\text{such that }
  \tau', f \cup \{ y \rightarrow x_{k+1} \} \models \psi_t(\xs,y).
\end{aligned}
\]
\[
\begin{aligned}
  P_{4,\Omega}(\tau)
  \iff{}&
  \text{for every }
  x \in \operatorname{dom}(\tau)
  \text{ and for all }
  \alpha_1,\alpha_2 \in \ell_{\tau}(x),
  \\
  &\text{there exists a labelled $2$-type }
  \beta \in \Omega
  \text{ such that}
  \\
  &\type{\beta}{x_1} = \alpha_1,
  \qquad
  \type{\beta}{x_2} = \alpha_2,
  \\
  &\beta \models x_1 E_1 x_2,
  \qquad
  \ell_{\beta}(x_1)=\ell_{\beta}(x_2)=\ell_{\tau}(x).
\end{aligned}
\]
\end{lemma}

In the following, we write $\AAA$ for the set of $1$-types over $\sigma$,
and $m$ for the number of variables occurring in~$\phi$.  

\subparagraph*{Completeness.}
Let $\str{A} \models \phi$ be a model as in Lemma~\ref{l:propermodels-full}.
Define $\Omega$ as the set of labelled $k$-types with $k \le m$ \emph{realised} in $\str{A}$;
that is, for every $k \in [0,m]$ and every distinct elements $a_1,\dots,a_k \in A \setminus \Cons$, include in $\Omega$ the labelled $k$-type $(\tau,\ell)$, where
\begin{itemize}
  \item $\tau = \type{\str{A}}{a_1,\dots,a_k}$,
  \item for every $x \in \operatorname{dom}(\tau)$,
    \[\ell(x) = \big\{\, \type{\str{A}}{b} \bigmid b \in \absclass{\str{A}}{1}{x^*} \setminus \Cons \,\big\},\]
    where $x^* = x$ if $x \in \Cons$, and $x^* = a_i$ if $x = x_i$ for some $i \in [k]$.
\end{itemize}

One can verify that the conditions of Lemma~\ref{l:decision-procedure} hold for the constructed set $\Omega$.

\subparagraph*{Soundness.}
Suppose that $\Omega$ is a set of labelled types over $\sigma$ as in Lemma~\ref{l:decision-procedure}.
To prove that $\phi$ is satisfiable, we construct a chain of finite structures
\[ \str{B}_1 \subseteq \str{B}_2 \subseteq \str{B}_3 \subseteq \dots, \quad\quad \text{where } \str{B}_{n+1} \,\restr\, B_n = \str{B}_n \text{ for all } n \ge 1. \]
The natural limit structure $\str{B}_{\omega} \eqdef \bigcup_{n \ge 1}\str{B}_n$ will be a model of $\phi$.

During the construction,
we define a sequence of functions $(\ell_n)_{n \in \N}$,
where each function $\ell_n \colon B_n \rightarrow 2^{\AAA}$ assigns to every element of the domain of $\str{B}_n$ a set of $1$-types.
The label $\ell_n(a)$ serves as a \emph{declaration} of which $1$-types are admitted in the $E_1$-class of $a$.
The following conditions shall hold during the construction:
\begin{itemize}
  \item \( \{ \type{\str{B}_n}{b} \mid b \in \absclass{\str{B}_n}{1}{a} \setminus \Cons \} \subseteq \ell_n(a) \) for all $a \in B_n$.
  \item $\ell_n(a) = \ell_n(b)$ for all $(a,b) \in E^{\str{B}_n}_1$.
  \item Let $\bar{b} = (b_1,\dots,b_k)$ be a $k$-tuple of distinct elements from $B_n \setminus \Cons$.
    If the $k$-type $\type{\str{B}_n}{\bar{b}}$ is guarded,
    then there exists a labelled $k$-type $(\tau,\ell_{\tau}) \in \Omega$ such that:
    \begin{itemize}
        \item $\type{\str{B}_n}{\bar{b}} = \tau$.
        \item For each $i \in [k]$, $\ell_{\tau}(x_i) = \ell_{n}(b_i)$.
        \item For each $c \in \Cons$, $\ell_{\tau}(c) = \ell_{n}(c)$.
    \end{itemize}
\end{itemize}

\smallskip\noindent\emph{Base case.}
Let $(\tau,\ell_{\tau}) \in \Omega$ be a labelled $k$-type with $k > 0$;
such a $k$-type exists by $P_{0,\Omega}$.
Define $\str{B}_1$ as the structure with domain $\{ b\} \cup \Cons$ such that the element $b$ realises $\tau$.
The function $\ell_1$ is then defined accordingly to $\ell_{\tau}$.

\smallskip\noindent\emph{Inductive step.}
Let $n \ge 1$.
Assume that $\str{B}_n$ and $\ell_n$ have already been defined.
Choose in a \emph{fair way} an \emph{unsatisfied} pair $(f,t)$, where $f\colon \xs \rightarrow B_n$ and $t \in \cI$;
that is, a pair such that $\str{B}_n,f \models G_t$ but $\str{B}_n,f \not\models \exists y\;\psi_t$.

Let $\bar{b} = (b_1,\dots,b_k)$ enumerate $f(\xs) \setminus \Cons$.
W.l.o.g.~$\{ b_1,\dots,b_k \} = f(\FreeVars(G_t)) \setminus \Cons$,
that is, we discard elements not in the scope of the guard $G_t$ under $f$.
Construct a labelled $k$-type $(\tau,\ell_\tau)$, where
\begin{itemize}
    \item $\tau = \type{\str{B}_n}{b_1,\dots,b_k}$,
    \item for each $i \in [k]$, $\ell_\tau(x_i) = \ell_n(b_i)$,
    \item for each $c \in \Cons$, $\ell_\tau(c) = \ell_n(c)$.
\end{itemize}

Define an assignment $g\colon \xs \rightarrow \operatorname{dom}(\tau)$ in parallel to $f$:
\[
    g(x) =
    \begin{cases}
        \, x_i \quad & \text{if} \quad f(x) = b_i, \\
        \, f(x) \quad & \text{if} \quad f(x) \in \Cons.
    \end{cases}
\]

From the inductive invariant, we have $(\tau,\ell_\tau) \in \Omega$.
Then $P_{3,\Omega}$ yields a labelled $(k{+}1)$-type $(\tau',\ell_{\tau'}) \in \Omega$ canonically extending $(\tau,\ell_{\tau})$ such that $\tau', g \cup \{ y \mapsto x_{k+1} \} \models \psi_t$ holds.

We now construct the structure $\str{B}_{n+1}$ by extending $\str{B}_n$ with a fresh element $b_{k+1}$.
\begin{itemize}
  \item Assign the $(k{+}1)$-type of $(b_1,\dots,b_{k+1})$:
  \[\type{\str{B}_{n+1}}{b_1,\dots,b_{k+1}} \eqdef \tau'. \]
Since $\tau'$ canonically extends $\tau$, this step does not modify the $k$-type $\type{\str{B}_n}{b_1,\dots,b_k}$, i.e., $\str{B}_{n+1}~\restr~B_n = \str{B}_n$.
  \item If there exists $a \in \{ b_1,\dots,b_k \} \cup \Cons$ such that $\str{B}_{n+1} \models a E_1 b_{k+1}$,
        then the corresponding transitive connections need to be added.
        For every $b \in \absclass{\str{B}_n}{1}{a} \setminus (\bar{b} \cup \Cons)$,
        by $P_{4,\Omega}$, there exists a labelled $2$-type $(\beta,\ell_{\beta}) \in \Omega$ such that:
        \[
        \begin{aligned}
        &\type{\beta}{x_1} = \type{\str{B}_{n+1}}{b},
        \qquad
        \type{\beta}{x_2} = \type{\str{B}_{n+1}}{b_{k+1}},
        \\
        &\beta \models x_1 E_1 x_2,
        \qquad
        \ell_{\beta}(x_1) = \ell_{\beta}(x_2) = \ell_{\tau'}(x_{k+1}).
        \end{aligned}
        \]
      Set $\type{\str{B}_{n+1}}{b,b_{k+1}} \eqdef \beta$. 
\end{itemize}

We complete the inductive step by declaring
\[ \ell_{n+1} \eqdef \ell_n \cup \{ b_{k+1} \mapsto \ell_{\tau'}(x_{k+1}) \}.\]

\subparagraph*{Decision procedure.}

Algorithm~\ref{algos2} decides the satisfiability criterion of Lemma~\ref{l:decision-procedure}.
It follows the same principles as Algorithm~\ref{algos} from \autoref{sec:GF3-decision-procedure}, with minor modifications necessitated by the presence of constants and an unbounded number of variables.

\begin{algorithm}[t]
  \caption{Decide satisfiability for equality-free $\GF[\KEQ{1}]$}\label{algos2}
  \KwIn{A normal-form sentence $\phi$ over a signature $\sigma$.}
  \KwOut{\textsc{Sat} if $\phi$ is satisfiable, and \textsc{UnSat} otherwise.}
  $m \gets$ the number of variables in $\phi$\\
  \ForEach{\rm labelled $0$-type $\pi$ over $\sigma$}{
    $\Omega \gets \{$all labelled $k$-types $\tau$ over $\sigma$ such that $0 \le k \le m$ and $\tau \models \pi\}$\;
    \While{\rm there exists $\tau \in \Omega$ violating one of~$P_{1,\Omega}(\tau),\dots,P_{4,\Omega}(\tau)$}{
      choose such a $\tau$\;
      $\Omega \gets \Omega \setminus \{\tau\}$\;
    }
    \lIf{\rm $P_{0,\Omega}$ holds}{\Return \textsc{Sat}}
  }
  \Return \textsc{UnSat}\;
\end{algorithm}

We now analyse its running time.
The procedure runs in time polynomial in the number of labelled types over at most $m$ variables.
Each labelled type is a pair $(\tau,\ell_\tau)$, where $\tau$ is a $k$-type with $k \leq m$ and
\(\ell_\tau \colon \operatorname{dom}(\tau) \to 2^{\AAA}.\)
The number of $k$-types is doubly exponential in~$|\phi|$.
Moreover, this bound is already tight for $k=1$ in the presence of constants and unbounded-arity relation symbols.
Consequently, the number of possible labellings $\ell_\tau$ is triply exponential in~$|\phi|$.
It follows that the total number of labelled types is triply exponential in~$|\phi|$, and therefore the procedure runs in deterministic triply-exponential time.

\subsection{Complexity of Constant-Free and Finite-Variable Fragments}\label{sec:restricted-subfragments}

We now show that constant-free and fixed-variable subfragments of \GF$[\KEQ{K}]$ admit lower complexity bounds,
namely, $(K{+}1)$-\ExpTime{} instead of $(K{+}2)$-\ExpTime{}.

The bottleneck in the running time of Algorithm~\ref{algos2} comes from the number of possible labellings $\operatorname{dom}(\tau) \rightarrow 2^{\AAA}$.
When the number of $1$-types (= the size of $\AAA$) becomes singly exponential instead of doubly exponential, the running time decreases by one exponential level.
In what follows, we show that this is indeed the case.   

Without constants the size of $\AAA$ decreases to singly exponential, yielding immediately the $2$-\ExpTime{} complexity for satisfiability of $\GF[\KEQ{1}]$.
To extend this bound to $\GF[\KEQ{K}]$ with $K \ge 2$,
we use the reduction of \autoref{l:GF3-reduction};
the more succinct reduction of \autoref{l:succinct-reduction} is not applicable for constant-free fragments, as it relies on the succinct axiomatisation of $2$-numerals (described in~\autoref{appendix:GF3-lower}) that requires the presence of constants.

The fixed variable case follows by taking a \emph{$\phi$-relative} set of $1$-types, denoted $\AAA_{\phi}$, instead of the entire space $\AAA$. The set $\AAA_{\phi}$ is defined as follows.

We first introduce an equivalence relation $\sim_{\phi}$ on the set $\AAA$.
Let $V$ denote the set of variables occurring in $\phi$.
Set $\alpha_1 \sim_{\phi} \alpha_2$ if, for every atom $\gamma$ occurring in $\phi$ and every assignment of variables $f\colon V \rightarrow \{ x_1 \} \cup \Cons$,
it holds that
\[ \alpha_1,f \models \gamma \iff \alpha_2,f \models \gamma. \]
Then $\AAA_{\phi}$ includes one representative from every equivalence class in $\AAA/_{\sim_{\phi}}$.
The correctness of this construction is immediate: $\phi$ cannot distinguish any two $\sim_{\phi}$-related $1$-types.
Since $V$ is fixed, the resulting set is of singly exponential size: 
\[|\AAA_{\phi}| = |\AAA/_{\sim_{\phi}}| = 2^{|\phi|\cdot(|\Cons| + 1)^{|V|}}.\]

\section{Undecidability: Two Equivalences without Nesting Constraints}\label{sec:GF-undec}

In this section, we establish the undecidability of the constant-free, equality-free fragment $\GF^3$ extended with two distinguished binary predicates $E$ and $F$, each interpreted as an equivalence relation, but without any nesting condition between them (Theorem~\ref{thm:GFundec}).

More precisely, \autoref{sec:undec-sat} proves undecidability of the general satisfiability problem, and \autoref{sec:undec-finsat} establishes undecidability of finite satisfiability.

\subsection{General Satisfiability}\label{sec:undec-sat}

The reduction is from the \emph{infinite tiling problem}, which is known to be co-\RE{}-complete~\cite{Ber66,GK72}.

A tiling instance is a tuple $\cT = \langle \cC, c_{0}, \cH, \cV \rangle$, where $\cC$ is a non-empty, finite set of \emph{colours},  
$c_{0} \in \cC$ is the \emph{initial} colour, and $\cH, \cV \subseteq \cC \times \cC$ are the sets of \emph{horizontal} and \emph{vertical} constraints, respectively.  
The infinite tiling problem is to decide if the instance $\cT$ is \emph{solvable}, i.e.,\ if there exists a \emph{tiling} function $f\colon \N \times \N \to \cC$ such that the following holds:
\begin{itemize}
  \item The initial condition holds: $f(0,0) = c_0$.
  \item The horizontal constraints hold: for all $i,j \in \N$, we have
	\[\langle f(i,j), f(i + 1,j) \rangle \in \cH. \]
  \item The vertical constraints hold: for all $i,j \in \N$, we have 
  \[ \langle f(i,j), f(i,j + 1) \rangle \in \cV. \]
\end{itemize}

We first describe the two-way infinite grid $\Z \times \Z$;
then define the grid $\N \times \N$;
and finally express the tiling conditions.

The standard model for the grid $\Z \times \Z$ is the structure $\str{G}_\Z$ over the signature $\sigma = \{ E, F, C, R \}$,
where
\begin{itemize}
  \item $E$ and $F$ are the distinguished equivalence symbols;
  \item $C$ and $R$ are auxiliary unary symbols marking even columns and even rows, respectively.
\end{itemize}

The domain of $\str{G}_\Z$ is the set $\Z \times \Z$, and the interpretation of symbols is as follows:
\begin{itemize}
  \item $E^{\str{G}_Z}$ is the equivalence closure of the set
  \[ \big\{ \big\langle(2 i,2 j), (2 i + d_1, 2 j + d_2)\big\rangle \bigmid i,j \in \Z, \, d_1,d_2 \in \{ 0, 1 \} \big\}. \]
  \item $F^{\str{G}_Z}$ is the equivalence closure of the set
  \[ \big\{ \big\langle(2 i,2 j), (2 i - d_1, 2 j - d_2)\big\rangle \bigmid i,j \in \Z, \, d_1,d_2 \in \{ 0, 1 \} \big\}. \]
  \item $C^{\str{G}_Z} \eqdef \{ (i,j) \in \Z^2 \mid i \text{ is even}\}$.
  \item $R^{\str{G}_Z} \eqdef \{ (i,j) \in \Z^2 \mid j \text{ is even}\}$.
\end{itemize}

Note that $E^{\str{G}_Z}$ and $F^{\str{G}_Z}$ partition the domain into $2 \times 2$ squares; where every $E$-class intersects precisely four distinct $F$-classes: namely, every $E$-class contains precisely one element for every combination of literals: $\{ C, R \}$, $\{ \neg C, R \}$, $\{ C, \neg R \}$, and $\{ \neg C, \neg R \}$.
The analogous property holds also when the roles of $E$ and $F$ are swapped.
These properties allow us to capture horizontally and vertically neighbouring elements as follows.
\begin{itemize}
  \item The formula $\lambda_H(x,y)$ defines that $y$ is a horizontal successor of $x$ in the structure $\str{G}_\Z$: 
	\begin{flalign*}
	\lambda_H(x,y) \eqdef &\Big(\big(R(x) \leftrightarrow R(y)\big) \wedge C(x) \wedge \neg C(y) \wedge x E y \wedge \neg x F y \Big)\;\vee \\
	&\Big(\big(R(x) \leftrightarrow R(y)\big) \wedge \neg C(x) \wedge C(y) \wedge \neg x E y \wedge x F y \Big).
	\end{flalign*}
  \item The formula $\lambda_V(x,y)$, specifying that $y$ is a vertical successor of $x$, is defined via symmetry:
	\begin{flalign*}
	\lambda_V(x,y) \eqdef &\Big(\big(C(x) \leftrightarrow C(y)\big) \wedge R(x) \wedge \neg R(y) \wedge x E y \wedge \neg x F y \Big)\;\vee \\
	&\Big(\big(C(x) \leftrightarrow C(y)\big) \wedge \neg R(x) \wedge R(y) \wedge \neg x E y \wedge x F y \Big).
	\end{flalign*}
\end{itemize}

Define $\str{G}_\N \,\eqdef\, \str{G}_\Z ~\restr~ \N \times \N$.
In the following, we axiomatise its key properties.

First, every element has a horizontal and vertical successor:
\begin{flalign}
\label{eq:undec-first}
\forall x\;\exists y\; \lambda_H(x,y) \quad\wedge\quad \forall x\;\exists y\;\lambda_V(x,y) 
\end{flalign}
The horizontal and vertical successors properly commutes:
\begin{flalign}\label{eq:undec-commutes}
  \bigwedge_{G \in \{ E,F \}}\forall x_1,x_2\; \Big(&\big(G(x_1,x_2) \wedge \lambda_V(x_1,x_2)\big) \rightarrow \nonumber \\
  &\exists y_1,y_2\;\big(\lambda_H(x_1,y_1) \wedge \lambda_H(x_2,y_2) \wedge \lambda_V(y_1,y_2)\big)\Big)
\end{flalign}
Since $\lambda_V(x_1,x_2) \models x_1 E x_2 \vee x_1 F x_2$,
the predicate $G \in \{ E, F\}$ is in fact a vacuous guard; added only to comply with the syntax of \GF{}.

Although formula~\eqref{eq:undec-commutes} uses $4$ variables, it is equivalent to one with $3$ variables: there are no interactions between pairs of variables $(x_1,y_2)$ and $(x_2,y_1)$.
Recall that we have already employed a similar trick to rewrite formulas~\eqref{eqn:three-variable-trik1}--\eqref{eqn:three-variable-trik2} in \autoref{sec:GF3-large-counters}.

We now express that every $E$-class (resp. $F$-class) touches at most four neighbouring $F$-classes (resp. $E$-classes).
More precisely, we enforce that the respective north-east/north-west/south-west/south-east neighbour, if exists, is unique:
\begin{flalign}
  &\forall x_1,x_2\;\big(\big(x_1 E x_2 \wedge \mu(x_1,x_2)\big) \rightarrow x_1 F x_2\big) \label{eq:uses-mu1} \\
  &\forall x_1,x_2\;\big(\big(x_1 F x_2 \wedge \mu(x_1,x_2)\big) \rightarrow x_1 E x_2\big) \label{eq:uses-mu2} \\
  &\forall x_1,x_2\;\big(\big(x_1 E x_2 \wedge x_1 F x_2\big) \rightarrow \mu(x_1,x_2) \big)
  \label{eq:undec-last}
\end{flalign}
Above, $\mu(x_1,x_2)$ abbreviates the formula
\[ \mu(x_1,x_2) \eqdef \big(R(x_1) \leftrightarrow R(x_2)\big) \wedge \big(C(x_1) \leftrightarrow C(x_2)\big). \]

The axioms \eqref{eq:undec-first}--\eqref{eq:undec-last} entail the existence of a homomorphic copy of $\str{G}_{\N}$:
\begin{claim}\label{claim:infinite-grid}
  Let $\Theta_\cG$ be the conjunction of~\eqref{eq:undec-first}--\eqref{eq:undec-last}.
  Then $\Theta_\cG \models \str{G}_\N$.
  Moreover, for every model $\str{A}\models\Theta_\cG$ and for every element $a \in C^{\str{A}} \cap R^{\str{A}}$, there exists a homomorphism $h_a\colon \str{G}_\N \rightarrow \str{A}$ such that $h(0,0) = a$.
\end{claim}

Having the sentence $\Theta_\cG$, we express the infinite tiling problem:
fix an instance $\cT = \langle \cC, c_{0}, \cH, \cV \rangle$.
Introduce for each colour $c\in \cC$ a unary predicate $P_c$.
We express the initial condition and the uniqueness of colours:
\begin{flalign}\label{eqn:inf-first-tiling}
  &\exists x\;P_{c_0}(x) \; \wedge\; \forall x\;\bigvee_{c \in C} P_c(x)
\end{flalign}
\begin{flalign}\label{eqn:specjalnejtroski}
  &\forall x_1,x_2\;\Big(\big(x_1 E x_2 \wedge x_1 F x_2\big) \rightarrow \bigwedge_{\substack{c,c' \in \cC \\ c \neq c'}} \big(\neg P_c(x_1) \vee \neg P_{c'}(x_2)\big)\Big)
\end{flalign}
The hereditary conditions are easy to express:
\begin{flalign}
  \bigwedge_{G \in \{ E,F \}} \forall x_1,x_2\; \Big(&G(x_1,x_2) \rightarrow \nonumber \\
  &\Big(\lambda_H(x_1,x_2) \rightarrow \bigvee_{\langle c,c' \rangle \in \cH} \big(P_c(x_1) \wedge P_{c'}(x_2)\big)\Big) \; \wedge \nonumber \\
  &\Big(\lambda_V(x_1,x_2) \rightarrow \bigvee_{\langle c,c' \rangle \in \cV} \big(P_c(x_1) \wedge P_{c'}(x_2)\big)\Big)\Big)\label{eqn:inf-last-tiling}
\end{flalign}

\begin{claim}\label{claim:infinite-tiling-sat}
  Let $\Theta_{\mathcal{T}}$ be the conjunction of~\eqref{eqn:inf-first-tiling}--\eqref{eqn:inf-last-tiling}.
  Then $\cT$ is solvable if and only if $\Theta_{\cG} \wedge \Theta_{\mathcal{T}}$ is satisfiable.
\end{claim}

\subsection{Finite Satisfiability}\label{sec:undec-finsat}

The undecidability of finite satisfiability follows from the \emph{toroidal-grid tiling problem};
known to be \RE{}-complete.
This problem is essentially the infinite tiling problem but with added constraint on the tiling function $f\colon \N \times \N \rightarrow \cC$. Namely, we look for solutions such that there exist $n,m \ge 1$ satisfying
\[ \text{$f(i,j) = f(i + n, j) = f(i,j + m)$ for all $i,j \in \N$.} \]

Observe that whenever the tiling function $f$ satisfies the constraint for $n$ and $m$, then it also does so for $2n$ and $2m$.
Let $\str{G}_{2n,2m}$ denote the quotient structure of $\str{G}_\Z$ via
\[ (i,j) \mapsto (i \operatorname{mod} 2n,j \operatorname{mod} 2m).\]

The sentence $\Theta_\cG$, defined in Claim~\ref{claim:infinite-grid}, models also the \emph{toroidal grid} $\str{G}_{2n,2m}$.
\begin{claim}\label{claim:toroidal-grid}
  It holds that $\Theta_\cG \models \str{G}_{2n,2m}$ for all $n,m \in \N$.
  Moreover, for every finite model $\str{A}\models\Theta_\cG$ and every element $a \in C^{\str{A}} \cap R^{\str{A}}$, there exists a homomorphism $h_a\colon \str{G}_{2n,2m} \rightarrow \str{A}$ for some $n,m \ge 1$ such that $h(0,0) = a$.
\end{claim}

The rest of the proof follows by the same argument as in the case of general satisfiability.

\section{Triguarded Fragment with Nested Equivalences}\label{sec:TGF-complexity}

\begin{proposition}
  The equality-free \TGF{}$[\EQ]$ enjoys the finite model property
  and has a \TOWER{}-complete satisfiability problem. 

  For every $K \ge 1$, the satisfiability problem for the equality-free $\TGF[\KEQ{K}]$ is $(K{+}2)$-\NExpTime{}-complete.
  This complexity bound decreases to $(K{+}1)$-\NExpTime{}-complete whenever either
  \begin{itemize}
    \item constants are disallowed, or
    \item the number of variables $m \ge 3$ is fixed.
  \end{itemize}
\end{proposition}

\subparagraph*{Upper bound.}

In \autoref{sec:FMP-full}, we established the finite model property for $\GF[\KEQ{K}]$ by first embedding $\GF[\KEQ{K}]$ into $\GF[\KEQ{(K{+}1)}]$ via quantifier relativisation by $E_{K+1}$, and then reducing to the fragment $\GFU$ (Lemma~\ref{l:GF3-reduction}), which has the same expressive power as \TGF{}.
The key observation underlying the reduction is that $E_{K+1}$ can be replaced by the universal relation~$U$.

Now observe that $\TGF[\KEQ{K}]$ likewise embeds into $\GF[\KEQ{(K{+}1)}]$ by the same relativisation argument.
Consequently, Lemma~\ref{l:GF3-reduction} also yields a reduction from $\TGF[\KEQ{K}]$ to~\TGF.
It follows that $\TGF[\KEQ{K}]$ enjoys the finite model property with the same upper bounds on minimal model size.
Since these bounds match the claimed complexity bounds, the result follows by a straightforward guess-and-verify argument.

\subparagraph*{Lower bound.}

We now establish matching lower bounds by adapting the constructions from \autoref{sec:GF3-lower}.
We focus on the constant-free, equality-free fragment $\TGF^3[\KEQ{K}]$.
The lower bound for equality-free $\TGF[\KEQ{K}]$ is then obtained by replacing the standard formulas encoding integers from $[0,2^n-1]$ with the succinct axiomatisation of integers from the larger range $[0,2^{2^n}-1]$ developed in \autoref{sec:lifting-lower-bounds}.

The reduction proceeds via the \emph{$(K{+}1)$-exponential tiling problem}, known to be $(K{+}1)$-\NExpTime-hard.
An instance of this problem is a tuple $\cT = \langle \cC, c_{0}, \cH, \cV, n \rangle$, where $\cC$ is a non-empty finite set of \emph{colours},  
$c_{0} \in \cC$ is the \emph{initial} colour, $\cH, \cV \subseteq \cC \times \cC$ are the sets of \emph{horizontal} and \emph{vertical} constraints, respectively, and $n \in \N$ is given in unary.  
Write $m = \tet{K+1}{n}$.
The problem is to decide if the instance $\cT$ is \emph{solvable}, that is, if there exists a \emph{tiling} function $f\colon [0,m{-}1] \times [0,m{-}1] \to \cC$ such that the following holds:
\begin{itemize}
  \item The initial condition holds: $f(i,j) = c_0$ for some $i,j \in [0,m{-}1]$.
  \item The horizontal constraints hold: for all $i,j \in [0,m{-}1]$, we have
  \[ \langle f(i,j), f(i',j) \rangle \in \cH,\]
  where $i' \equiv (i+1)\operatorname{mod}m$.
  \item The vertical constraints hold: for all $i,j \in [0,m{-}1]$, we have
  \[\langle f(i,j), f(i,j') \rangle \in \cV,\]
  where $j' \equiv (j+1)\operatorname{mod}m$.
\end{itemize}

Let $K \ge 1$ and $n \ge 1$, and write $m = \tet{K+1}{n}$.
We describe a theory of the $m \times m$ toroidal grid.
In $\TGF^3[\KEQ{K}]$, the distinguished predicates are $E_1,\dots,E_K$.
We axiomatise the additional predicate $E_{K+1}$ as the universal relation:
\begin{flalign}\label{eqn:exp-grid-first}
 \forall x,y\;x E_{K+1} y
\end{flalign}

Let $\phi_{K+1,n}^{(1)},\dots,\phi_{K+1,n}^{(K+1)}$ be sentences as constructed in \autoref{sec:GF3-large-counters}.
To axiomatise the coordinate axes, we rewrite the last sentence $\phi_{K+1,n}^{(K+1)}$ into two sentences $\phi_X^{(K+1)}$ and $\phi_Y^{(K+1)}$.

Introduce fresh symbols $B_{\diamond}$, $C_{\diamond}$, $S_{\diamond}$, $Z_{\diamond}$, and $Q_{\diamond}$ for $\diamond \in \{ X, Y\}$.
Define $\phi_X^{(K+1)}$ as the conjunction of sentences isomorphic to~\eqref{phik-first}--\eqref{phik-pre-last} but with the predicates $B_{K+1}$, $C_{K+1}$, $S_{K+1}$, $Z_{K+1}$, and $Q_{K+1}$ replaced by their respective counterparts for $\diamond = X$.
Let $\phi_Y^{(K+1)}$ be defined as $\phi_X^{(K+1)}$ but for $\diamond = Y$.
Note that in $\phi_X^{(K+1)}$ and $\phi_Y^{(K+1)}$ we omitted the last conjunct~\eqref{phik-last} of the original sentence $\phi_{K+1,n}^{(K+1)}$.
This is as expected, since~\eqref{phik-last} states that there is at most one $E_K$-class for every $(K{+}1)$-numeral.
Naturally, for the grid, we require at most one $E_K$-class for every pair of $(K{+}1)$-numerals:
\begin{flalign}\label{eqn:exp-grid-second}
  \forall x_1,x_2\;\big(\big(Q_X(x_1,x_2) \wedge Q_Y(x_1,x_2)\big) \rightarrow x_1 E_K x_2\big)
\end{flalign}
We also require proper horizontal and vertical successors:
\begin{flalign}
  &\forall x\;\exists y\;\big(S_X(x,y) \wedge Q_Y(x,y)\big)\label{eqn:exp-grid-mid} \\
  &\forall x\;\exists y\;\big(Q_X(x,y) \wedge S_Y(x,y)\big)\label{eqn:exp-grid-last}
\end{flalign}
Finally, define
\[ \Theta_{K,n} \;\eqdef\; \phi_{K+1,n}^{(1)} \wedge \dots \wedge \phi_{K+1,n}^{(K)} \wedge \phi^{(K+1)}_X \wedge \phi^{(K+1)}_Y \wedge \phi_{\eqref{eqn:exp-grid-first}} \wedge \phi_{\eqref{eqn:exp-grid-second}} \wedge \phi_{\eqref{eqn:exp-grid-mid}} \wedge \phi_{\eqref{eqn:exp-grid-last}} \]

\begin{claim}\label{claim:grid-invariant}
  The sentence $\Theta_{K,n}$ is satisfiable.
  Moreover, for every model $\str{A}\models \Theta_{K,n}$, one can choose functions
  \[ \num_X,\num_Y\colon A \longrightarrow [0,\tet{K+1}{n}-1]\]
  such that for $\diamond \in \{X,Y\}$ and $a,b \in A$ the following holds:
  \begin{itemize}
    \item $\str{A} \models a E_K b$ \, iff \, $\num_{X}(a) = \num_{X}(b)$ and $\num_{Y}(a) = \num_{Y}(b)$.
    \item $\str{A} \models S_{\diamond}(a,b)$ \, iff \,
    \(\num_{\diamond}(a) + 1 \equiv \num_{\diamond}(b) \mod \tet{K+1}{n}.\)
    \item $\str{A} \models Z_{\diamond}(a)$ \, iff \, $\num_{\diamond}(a) = 0$.
    \item $\str{A} \models Q_{\diamond}(a,b)$ \, iff \, $\num_{\diamond}(a) = \num_{\diamond}(b)$.
  \end{itemize}
\end{claim}

Now, it is routine to express the required tiling conditions:
fix an instance $\cT = \langle \cC, c_{0}, \cH, \cV, n \rangle$.
Introduce for each colour $c\in \cC$ a unary predicate $P_c$,
and express the initial condition and the uniqueness of colours:
\begin{flalign}\label{eqn:first-tiling}
  &\exists x\;P_{c_0}(x) \; \wedge\; \forall x\;\bigvee_{c \in C} P_c(x)
\end{flalign}
\begin{flalign}\label{eqn:specjalnejtroski}
  &\forall x_1,x_2\;\Big(x_1 E_K x_2 \rightarrow \bigwedge_{\substack{c,c' \in \cC \\ c \neq c'}} \big(\neg P_c(x_1) \vee \neg P_{c'}(x_2)\big)\Big)
\end{flalign}
Finally, express the hereditary conditions:
\begin{flalign}
& \forall x_1,x_2\;\Big(\big(S_X(x_1,x_2) \wedge Q_Y(x_1,x_2)\big) \rightarrow \bigvee_{\langle c,c' \rangle \in \cH} \big(P_c(x_1) \wedge P_{c'}(x_2)\big)\Big) \\
& \forall x_1,x_2\;\Big(\big(Q_X(x_1,x_2) \wedge S_Y(x_1,x_2)\big) \rightarrow \bigvee_{\langle c,c' \rangle \in \cV} \big(P_c(x_1) \wedge P_{c'}(x_2)\big)\Big)\label{eqn:last-tiling}
\end{flalign}

\section{A Restricted Use of Equality: Capturing Nominals}\label{sec:nominals}

Our proof techniques allow a restricted use of equality: namely, the atoms $c = x$ and $c = c'$, where $x$ is a variable and $c$, $c'$ constants, pose no problem.

This follows from the general property of our model constructions:
models are built by cloning $k$-types from a fixed model of a sentence in a way that does not modify previous choices.
Since $k$-types include literals involving constants,
the local structure on constants and each \emph{guarded} tuple is isomorphically preserved.
Here, a tuple $\bar{a}$ is said to be guarded in a structure $\str{A}$ whenever $\str{A} \models R(\bar{b})$ for some relation symbol $R$ and a tuple $\bar{b}$ with $\bar{a} \subseteq \bar{b}$.

\end{document}